\documentclass[final,12pt]{article}
\usepackage{hyperref}
\usepackage{amssymb,amsmath,amsthm}
\usepackage[mathscr]{eucal}
\usepackage{tikz}
\usetikzlibrary{arrows,backgrounds,decorations.pathmorphing,decorations.pathreplacing,positioning,fit,shapes,chains}
\usepackage{cite}
\usepackage{enumerate}
\usepackage[conditional,light,first,bottomafter]{draftcopy}
\draftcopyName{DRAFT\space\today}{130}
\draftcopySetScale{65}
\usepackage[letterpaper,hmargin=3.7cm,vmargin=3.1cm]{geometry}
\usepackage{setspace}
\makeatletter
\renewcommand{\section}{\@startsection%
{section}%
{1}%
{0em}%
{1.7em}%
{1.2em}%
{\normalfont\large\centering\bfseries}}
\renewcommand{\@seccntformat}[1]%
{\csname the#1\endcsname.\hspace{0.5em}}

\numberwithin{equation}{section}
\renewcommand\appendix{\par
\setcounter{section}{0}%
\setcounter{subsection}{0}%
\setcounter{theorem}{0}
\setcounter{table}{0}
\setcounter{figure}{0}
\gdef\thetable{\Alph{table}}
\gdef\thefigure{\Alph{figure}}
\section*{Appendix}
\gdef\thesection{\Alph{section}}
\setcounter{section}{1}}

\newtheorem{theorem}{Theorem}[section]
\newtheorem{proposition}{Proposition}[section]
\newtheorem{lemma}{Lemma}[section]
\newtheorem{corollary}{Corollary}[section]
\theoremstyle{definition}
\newtheorem{definition}{Definition}
\newtheorem{remark}{Remark}

\newtheorem{convention}{Convention}

\newcommand{\abs}[1]{\left|#1\right|}
\newcommand{\norm}[1]{\left\|#1\right\|}
\newcommand{\inner}[2]{\left\langle#1,#2\right\rangle}
\newcommand{\cH}{\mathcal{H}}
\newcommand{\integers}{\mathbb{Z}}
\newcommand{\tb}[1]{\widetilde{\boldsymbol{#1}}}
\newcommand{\reals}{\mathbb{R}}
\newcommand{\complex}{\mathbb{C}}
\newcommand{\nats}{\mathbb{N}}
\newcommand{\pb}[1]{\boldsymbol{#1}}

\DeclareMathOperator{\spec}{spec}
\DeclareMathOperator{\Span}{span}

\DeclareMathOperator{\dom}{dom}

\DeclareMathOperator{\diag}{diag}
\begin{document}
\begin{titlepage}
\title{Inverse spectral analysis for a class of infinite band symmetric matrices
\footnotetext{%
Mathematics Subject Classification(2010):
34K29,  
47A75, 
47B36, 
70F17, 
}
\footnotetext{%
Keywords:
Inverse spectral problem;
Band symmetric matrices;
Spectral measure.
}\hspace{-5mm}
\thanks{%
Research partially supported by UNAM-DGAPA-PAPIIT IN105414
}%
}
\author{
\textbf{Mikhail Kudryavtsev}
\\
\small Department of Mathematics\\[-1.6mm]
\small Institute for Low Temperature Physics and Engineering\\[-1.6mm]
\small Lenin Av. 47, 61103\\[-1.6mm]
\small Kharkov, Ukraine\\[-1.6mm]
\small\texttt{kudryavtsev@onet.com.ua}
\\[2mm]
\textbf{Sergio Palafox}
\\
\small Departamento de F\'{i}sica Matem\'{a}tica\\[-1.6mm]
\small Instituto de Investigaciones en Matem\'aticas Aplicadas y en Sistemas\\[-1.6mm]
\small Universidad Nacional Aut\'onoma de M\'exico\\[-1.6mm]
\small C.P. 04510, M\'exico D.F.\\[-1.6mm]
\small \texttt{sergiopalafoxd@gmail.com}
\\[2mm]
\textbf{Luis O. Silva}
\\
\small Departamento de F\'{i}sica Matem\'{a}tica\\[-1.6mm]
\small Instituto de Investigaciones en Matem\'aticas Aplicadas y en Sistemas\\[-1.6mm]
\small Universidad Nacional Aut\'onoma de M\'exico\\[-1.6mm]
\small C.P. 04510, M\'exico D.F.\\[-1.6mm]
\small \texttt{silva@iimas.unam.mx} }
\date{}
\maketitle
\vspace{-6mm}
\begin{center}
\begin{minipage}{5in}
  \centerline{{\bf Abstract}} \bigskip This note deals with the direct
  and inverse spectral analysis for a class of infinite band symmetric
  matrices. This class corresponds to operators arising from
  difference equations with usual and \emph{inner} boundary
  conditions. We give a characterization of the spectral functions for
  the operators and provide necessary and sufficient conditions for a
  matrix-valued function to be a spectral function of the operators.
  Additionally, we give an algorithm for recovering the matrix from
  the spectral function. The approach to the inverse problem is based
  on the rational interpolation theory.
\end{minipage}
\end{center}
\thispagestyle{empty}
\end{titlepage}
\section{Introduction}
\label{sec:intro}
In this note, the direct and inverse spectral analysis of a class of
infinite symmetric band matrices, denoted $\mathcal{M}(n,\infty)$, is
carried out with emphasis in the inverse problems of characterization
and reconstruction. The matrices under consieration, defined in the
paragraphs below, arise from difference equations with initial and
left endpoint boundary conditions together with the so called
\emph{inner} boundary conditions. Inner boundary conditions are given
by degenerations of the diagonals (see the paragraphs below
Definition~\ref{def:matrices-degenerate}  and above
(\ref{eq:eigenvector})). Each matrix in $\mathcal{M}(n,\infty)$
generates uniquely a closed symmetric operator for which we give a
spectral characterization. More specifically, we provide necessary and
sufficient conditions for a matrix-valued function to be a spectral
function of the operators stemming from our class of matrices (see
Definition~\ref{def:class-sigma-infinite} and
Theorems~\ref{thm:sigma-unique} and \ref{thm:last-one}). As a
byproduct of the spectral analysis of the operators corresponding to
matrices in $\mathcal{M}(n,\infty)$ we find and if-and-only-if
criterion for degeneration in terms of the properties of polynomials
in a $L_2$ space (see Theorem~\ref{thm:non-degenerate-case}).

Although the inverse spectral problems for Jacobi matrices have been
studied extensively (see for instance
\cite{MR2263317,MR504044,MR2915295,MR1616422,MR0447294,MR0213379,
  MR0382314,MR549425,MR1463594,MR1436689,MR1247178} for the finite
case and
\cite{MR2998707,see-later-mr,MR1045318,MR1616422,MR499269,MR0221315,
  MR2305710,MR2438732} for the infinite case), works dealing with band
matrices non-necessary tridiagonal are not so abundant (see
\cite{MR629608,MR2533388,MR2592784,MR1668981,MR1699440,
  2014arXiv1409.3868K,MR636029,MR2110489,MR2432761} for the finite
case and \cite{MR2043894,MR2494240} for the infinite case).

Let $\cH$ be an infinite dimensional separable Hilbert space and fix
an orthonormal basis $\{\delta_k\}_{k=1}^\infty$ in it. We study the
symmetric operator $A$ whose matrix representation with respect to
$\{\delta_k\}_{k=1}^{\infty}$ is a symmetric band matrix which is
denoted by $\mathcal{A}$ (see \cite[Sec.~47]{MR1255973} for the
definition of the matrix representation of an unbounded symmetric
operator).

We assume that the matrix $\mathcal{A}$ has $2n+1$ band diagonals,
that is, $2n+1$ diagonals not necessarily zero. The band diagonals
satisfy the following conditions. The band diagonal farthest from the
main one, which is given by the diagonal matrix
$\diag\{d_k^{(n)}\}_{k=1}^{\infty}$, denoted by $\mathcal{D}_n$, is
such that, for some $m_1\in\nats$, all the numbers
$d_{1}^{(n)},\dots,d_{m_1-1}^{(n)}$ are positive and $d_{k}^{(n)}=0$
for all $k\geq m_1$ with
\begin{equation}
  \label{eq:first-m}
  m_1>1\,.
\end{equation}
It may happen that all the elements of the sequence
$\{d_k^{(n)}\}_{k\in\mathbb{N}}$ are positive which we
convene to mean that $m_1=\infty$.

Now, if $m_1<\infty$, then the elements
$\{d_{m_1+k}^{(n-1)}\}_{k=1}^{\infty}$ of the diagonal next to the
farthest, $\mathcal{D}_{n-1}$, behave in the same way as the elements
of $\mathcal{D}_n$, that is, there is $m_2$, satisfying
\begin{equation}
  \label{eq:second-m}
m_1<m_2\,,
\end{equation}
such that $d_{m_1+1}^{(n-1)},\dots,d_{m_2-1}^{(n-1)}>0$ and
$d_{k}^{(n-1)}=0$ for all $k\geq m_2$. Here, it is also possible that
$m_2=\infty$ in which case $d_k^{(n-1)}>0$ for all $k>m_1$.

We continue applying the same rule as long as $m_1,\dots,m_j$ are
finite. Thus, if $m_j<\infty$, there is $m_{j+1}$,
satisfying
\begin{equation}
  \label{eq:second-m}
m_j<m_{j+1}\,,
\end{equation}
such that $d_{m_j+1}^{(n-j)},\dots,d_{m_{j+1}-1}^{(n-j)}>0$ (here we
assume that $m_j+1<m_{j+1}$) and $d_{k}^{(n-j)}=0$ for all $k\geq
m_{j+1}$. If $m_j=\infty$, then $d_k^{(n-j)}>0$ for all
$k>m_{j}$. Eventually, there is $j_0\le n-1$ such that
$m_{j_0+1}=\infty$.

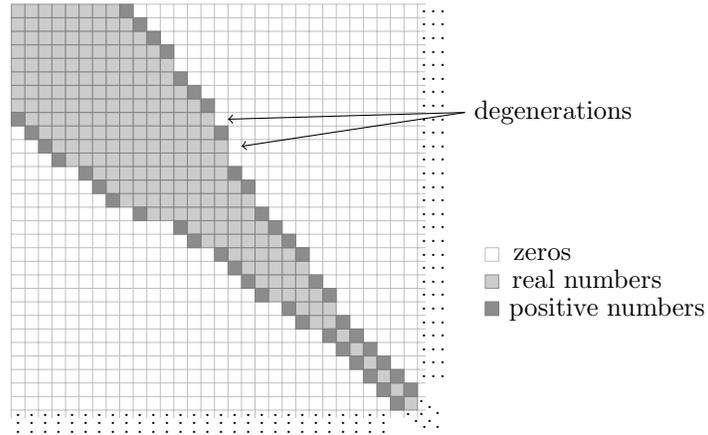
\begin{figure}[h]
\begin{center}
\begin{tikzpicture}[scale=.18]\footnotesize
 \pgfmathsetmacro{\xone}{0}
 \pgfmathsetmacro{\xtwo}{ 30.6}
 \pgfmathsetmacro{\yone}{-0.6}
 \pgfmathsetmacro{\ytwo}{30}
  \draw[step=1cm,gray,opacity=0.5,very thin] (\xone,\yone) grid (\xtwo,\ytwo);
  \draw[step=1cm,gray,opacity=0.5,very thin] (35,11) grid (36,12);
  \draw[step=1cm,gray,opacity=0.5,very thin] (35,9) grid (36,10);
  \draw[step=1cm,gray,opacity=0.5,very thin] (35,7) grid (36,8);
\draw(30.9,-.3)node[scale=.8]{$\ddots$};
\draw(30.9,.7)node[scale=.8]{$\ddots$};
\draw(29.9,-.3)node[scale=.8]{$\ddots$};
\draw(39.2,11.5)node[scale=1]{zeros};
\draw(42.5,9.7)node[scale=1]{real numbers};
\draw(44,7.5)node[scale=1]{positive numbers};
\draw(40,22)node[scale=1]{degenerations};
\draw[->] (33.5,22) -- (16,21.5);
\draw[->] (33.5,22) -- (17,19.5);

\begin{scope}
\foreach \x in {0,1,2,3,4,5,6,7,8,9,10,11,12,13,14,15,16,17,
18,19,20,21,22,23,24,25,26,27}
\draw(31.3,2.5+\x)node[scale=.8]{$\dots$}
;
\end{scope}

\begin{scope}
\foreach \x in {0,1,2,3,4,5,6,7,8,9,10,11,12,13,14,15,16,17,
18,19,20,21,22,23,24,25,26,27}
\draw(\x+.5,-.5)node[scale=.8]{$\vdots$}
;
\end{scope}

\begin{scope}
  \filldraw[thin,gray,opacity=.4] (35,9)
    rectangle (36,10)
 ;
  \filldraw[thin,gray,opacity=.9] (35,7)
    rectangle (36,8);
\end{scope}
\begin{scope}
\foreach \x in {0,1,2,3}
{
  \filldraw[thin,gray,opacity=.9] (0+\x, 21-\x)
    rectangle (1+\x,22-\x)
 ;
   \filldraw[thin,gray,opacity=.9] (8+\x, 30-\x)
     rectangle (9+\x,29-\x);}
\end{scope}
\begin{scope}
\foreach \x in {0,1,2,3,4}
{
  \filldraw[thin,gray,opacity=.4] (0+\x, 22-\x)
    rectangle (1+\x,23-\x)
 ;
   \filldraw[thin,gray,opacity=.4] (7+\x, 30-\x)
     rectangle (8+\x,29-\x);}
\end{scope}

\begin{scope}
\foreach \x in {5,6,7}
{
  \filldraw[thin,gray,opacity=.9] (0+\x, 22-\x)
    rectangle (1+\x,23-\x)
 ;
   \filldraw[thin,gray,opacity=.9] (7+\x, 30-\x)
     rectangle (8+\x,29-\x);}
\end{scope}
\begin{scope}
\foreach \x in {0,1,2,3,4,5,6,7,8}
{
  \filldraw[thin,gray,opacity=.4] (0+\x, 23-\x)
    rectangle (1+\x,24-\x)
 ;
   \filldraw[thin,gray,opacity=.4] (6+\x, 30-\x)
     rectangle (7+\x,29-\x);}
\end{scope}
\begin{scope}
\foreach \x in {9}
{
  \filldraw[thin,gray,opacity=.9] (0+\x, 23-\x)
    rectangle (1+\x,24-\x)
 ;
   \filldraw[thin,gray,opacity=.9] (6+\x, 30-\x)
     rectangle (7+\x,29-\x);}
\end{scope}
\begin{scope}
\foreach \x in {0,1,2,3,4,5,6,7,8,9,10}
{
  \filldraw[thin,gray,opacity=.4] (0+\x, 24-\x)
    rectangle (1+\x,25-\x)
 ;
   \filldraw[thin,gray,opacity=.4] (5+\x, 30-\x)
     rectangle (6+\x,29-\x);}
\end{scope}

\begin{scope}
\foreach \x in {0,1,2,3,4,5,6,7,8,9,10,11,12}
{
  \filldraw[thin,gray,opacity=.4] (0+\x, 25-\x)
    rectangle (1+\x,26-\x)
 ;
   \filldraw[thin,gray,opacity=.4] (4+\x, 30-\x)
     rectangle (5+\x,29-\x);}
\end{scope}

\begin{scope}
\foreach \x in {12,13}
{
  \filldraw[thin,gray,opacity=.9] (0+\x, 25-\x)
    rectangle (1+\x,26-\x)
 ;
   \filldraw[thin,gray,opacity=.9] (4+\x, 30-\x)
     rectangle (5+\x,29-\x);}
\end{scope}
\begin{scope}
\foreach \x in {0,1,2,3,4,5,6,7,8,9,10,11,12,13,14}
{
  \filldraw[thin,gray,opacity=.4] (0+\x, 26-\x)
    rectangle (1+\x,27-\x)
 ;
   \filldraw[thin,gray,opacity=.4] (3+\x, 30-\x)
     rectangle (4+\x,29-\x);}
\end{scope}

\begin{scope}
\foreach \x in {15,16,17,18}
{
  \filldraw[thin,gray,opacity=.9] (0+\x, 26-\x)
    rectangle (1+\x,27-\x)
 ;
   \filldraw[thin,gray,opacity=.9] (3+\x, 30-\x)
     rectangle (4+\x,29-\x);}
\end{scope}

\begin{scope}
\foreach \x in {0,1,2,3,4,5,6,7,8,9,10,11,12,13,14,15,16,17,18,19}
{
  \filldraw[thin,gray,opacity=.4] (0+\x, 27-\x)
    rectangle (1+\x,28-\x)
 ;
   \filldraw[thin,gray,opacity=.4] (2+\x, 30-\x)
     rectangle (3+\x,29-\x);}
\end{scope}

\begin{scope}
\foreach \x in {20,21}
{
  \filldraw[thin,gray,opacity=.9] (0+\x, 27-\x)
    rectangle (1+\x,28-\x)
 ;
   \filldraw[thin,gray,opacity=.9] (2+\x, 30-\x)
     rectangle (3+\x,29-\x);}
\end{scope}
\begin{scope}
\foreach \x in {0,1,2,3,4,5,6,7,8,9,10,11,12,13,14,15,
16,17,18,19,20,21,22}
{
  \filldraw[thin,gray,opacity=.4] (0+\x, 28-\x)
    rectangle (1+\x,29-\x)
 ;
   \filldraw[thin,gray,opacity=.4] (1+\x, 30-\x)
     rectangle (2+\x,29-\x);}
\end{scope}

\begin{scope}
\foreach \x in {23,24,25,26,27,28}
{
  \filldraw[thin,gray,opacity=.9] (0+\x, 28-\x)
    rectangle (1+\x,29-\x)
 ;
   \filldraw[thin,gray,opacity=.9] (1+\x, 30-\x)
     rectangle (2+\x,29-\x);}
\end{scope}

\begin{scope}
  \foreach \x in
  {0,1,2,3,4,5,6,7,8,9,10,11,12,13,14,15,16,17,18,19,
20,21,22,23,24,25,26,27,28,29}
  { \filldraw[thin,gray,opacity=.25] (0+\x, 29-\x) rectangle
    (1+\x,30-\x) ; \filldraw[thin,gray,opacity=.2] (0+\x, 29-\x)
    rectangle (1+\x,30-\x);}
\end{scope}

\end{tikzpicture}
\end{center}
\end{figure}

\begin{definition}
\label{def:matrices-degenerate}
For a natural number $n$, the set of matrices satisfying the above
properties with a given set of numbers $\{m_j\}_{j=1}^{j_0}$ is denoted by
$\mathcal{M}(n,\infty)$. 
\end{definition}

As long as $j\le j_0-1$, we say that the diagonal corresponding to
$\mathcal{D}_{n-j}$ undergoes degeneration at $m_{j+1}$. Note that the
diagonal corresponding to $\mathcal{D}_{n-j_0}$ do not
degenerate. Also, $j_0$ defines the number of degenerations that the
matrix $\mathcal{A}$ has.

\begin{remark}
\label{rem:tail-matrix}
Define the number $n_0:=n-j_0$.  Note that the ``tail'' of the matrix,
that is, the semi-infinite submatrix obtained by removing the first
$n_0+m_{j_0}-1$ columns and rows, has $2n_0+1$ diagonals and the
diagonal $\mathcal{D}_{n_0}$ has only positive numbers.
\end{remark}
An example of a matrix in $\mathcal{M}(3,\infty)$, when $m_1=3$ and
$m_2=5$, is the following.
\begin{equation}
  \label{eq:matrix-example}
\mathcal{A}=
\begin{footnotesize}
 \left(\begin{matrix}
    d^{(0)}_1&d^{(1)}_1&d^{(2)}_1&d^{(3)}_1&0&0&0&\dots\\[2.5mm]
    d^{(1)}_1&d^{(0)}_2&d^{(1)}_2&d^{(2)}_2&d^{(3)}_2&0&0&\\[1mm]
    d^{(2)}_1&d^{(1)}_2&d^{(0)}_3&d^{(1)}_3&d^{(2)}_3&0&0&\ddots\\[1mm]
    d^{(3)}_1&d^{(2)}_2&d^{(1)}_3&d^{(0)}_4&d^{(1)}_4&d^{(2)}_4&0&\ddots\\[1mm]
    0&d^{(3)}_2&d^{(2)}_3&d^{(1)}_4&d^{(0)}_5&d^{(1)}_5&0&\ddots\\[1mm]
    0&0&0&d^{(2)}_4&d^{(1)}_5&d^{(0)}_6&d^{(1)}_6&\ddots\\[1mm]
    0&0&0&0&0&d^{(1)}_{6}&d^{(0)}_7&\ddots\\[1mm]
    \vdots&&\ddots&\ddots&\ddots&\ddots&\ddots&\ddots
  \end{matrix}\right)\,.
\end{footnotesize}
\end{equation}
Here we say that the matrix $\mathcal{A}$ underwent a degeneration of
the diagonal $\mathcal{D}_3$ in $m_1=3$ and a degeneration of
$\mathcal{D}_2$ in $m_2=5$. And, note that $j_0=2$.


It is known that the dynamics of an infinite linear mass-spring system
(see Fig.~\ref{fig:0}) is characterized by the spectral properties of
a semi-infinite Jacobi matrix \cite{MR2998707,see-later-mr} when the
system is within the regime of validity of the Hooke law (see
\cite{MR2102477,mono-marchenko} for an explanation of how to obtain
the matrix from the mass-spring system in the finite case).  The
entries of the Jacobi matrix are determined by the masses and spring
constants of the system
\cite{MR2915295,MR2998707,see-later-mr,MR2102477,mono-marchenko}. The
movement of the mechanical system of Fig.~\ref{fig:0} is a
superposition of harmonic oscillations whose frequencies are the
square roots of absolute values of the Jacobi operator's eigenvalues.
\begin{figure}[h]
\begin{center}
\begin{tikzpicture}
  [mass1/.style={circle,draw=black!80,fill=black!13,thick,inner sep=0pt,
   minimum size=5mm},
   mass2/.style={circle,draw=black!80,fill=black!13,thick,inner sep=0pt,
   minimum size=3.7mm},
   mass3/.style={circle,draw=black!80,fill=black!13,thick,inner sep=0pt,
   minimum size=5.7mm},
   mass4/.style={circle,draw=black!80,fill=black!13,thick,inner sep=0pt,
   minimum size=5mm}, 
   mass5/.style={circle,draw=black!80,fill=black!13,thick,inner sep=0pt,
   minimum size=4mm}, 
   mass6/.style={circle,draw=black!80,fill=black!13,thick,inner sep=0pt,
   minimum size=5.2mm}, 
   mass7/.style={circle,draw=black!80,fill=black!13,thick,inner sep=0pt,
   minimum size=6mm}, 
   mass8/.style={circle,draw=black!80,fill=black!13,thick,inner sep=0pt,
   minimum size=5.2mm},
   mass9/.style={circle,draw=black!80,fill=black!13,thick,inner sep=0pt,
   minimum size=5.4mm},
   massn/.style={circle,draw=black!80,fill=black!13,thick,inner sep=0pt,
   minimum size=5.2mm},
   wall/.style={postaction={draw,decorate,decoration={border,angle=-45,
   amplitude=0.3cm,segment length=1.5mm}}}
   ]
  \node (massn) at (12.75,1) [massn] {};
  \node (mass9) at (11.5,1) [mass8] {};
  \node (mass8) at (10.25,1) [mass8] {};
  \node (mass7) at (9.0,1) [mass7] {};
  \node (mass6) at (7.75,1) [mass6] {};
  \node (mass5) at (6.5,1) [mass5] {};
  \node (mass4) at (5.25,1) [mass4] {};
  \node (mass3) at (4.0,1) [mass3] {};
  \node (mass2) at (2.75,1) [mass2] {};
  \node (mass1) at (1.5,1) [mass1] {};
\draw[decorate,decoration={coil,aspect=0.4,segment
  length=1.1mm,amplitude=0.7mm}] (0.5,1) -- node[below=4pt]
{} (mass1);
\draw[decorate,decoration={coil,aspect=0.4,segment
  length=1.4mm,amplitude=0.7mm}] (mass1) -- node[below=4pt]
{} (mass2);
\draw[decorate,decoration={coil,aspect=0.4,segment
  length=1.5mm,amplitude=0.7mm}] (mass2) -- node[below=4pt]
{} (mass3);
\draw[decorate,decoration={coil,aspect=0.4,segment
  length=1.1mm,amplitude=0.7mm}] (mass3) -- node[below=4pt]
{} (mass4);
\draw[decorate,decoration={coil,aspect=0.4,segment
  length=0.9mm,amplitude=0.7mm}] (mass4) -- node[below=4pt]
{} (mass5);
\draw[decorate,decoration={coil,aspect=0.4,segment
  length=1.4mm,amplitude=0.7mm}] (mass5) -- node[below=4pt]
{} (mass6);
\draw[decorate,decoration={coil,aspect=0.4,segment
  length=1.7mm,amplitude=0.7mm}] (mass6) -- node[below=4pt]
{} (mass7);
\draw[decorate,decoration={coil,aspect=0.4,segment
  length=0.8mm,amplitude=0.7mm}] (massn) -- node[below=4pt]
{} (13.75,1);
\draw[decorate,decoration={coil,aspect=0.4,segment
  length=1.1mm,amplitude=0.7mm}] (mass7) -- node[below=4pt]
{} (mass8);
\draw[decorate,decoration={coil,aspect=0.4,segment
  length=1.3mm,amplitude=0.7mm}] (mass8) -- node[below=4pt]
{} (mass9);
\draw[decorate,decoration={coil,aspect=0.4,segment
  length=1.7mm,amplitude=0.7mm}] (mass9) -- node[below=4pt]
{} (massn);
\draw[line width=.8pt,loosely dotted] (13.85,1) -- (14.25,1);
\draw[line width=.5pt,wall](0.5,1.7)--(0.5,0.3);
\end{tikzpicture}
\end{center}
\caption{Mass-spring system corresponding to a Jacobi matrix}
\label{fig:0}
\end{figure}
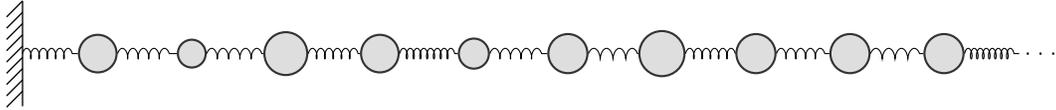
Analogously, one can deduce that a self-adjoint extension of the
minimal closed operator generated by a matrix in
$\mathcal{M}(n,\infty)$ models a linear mass-spring system where the
interaction extends to all the $n$ neighbors of each mass (cf. \cite[Appendix]{2014arXiv1409.3868K}). For
instance, if the matrix is in $\mathcal{M}(2,\infty)$ and no
degeneration of the diagonals occurs, viz. $m_1=\infty$, the corresponding
mass-spring system is given in Fig.~\ref{fig:1}.
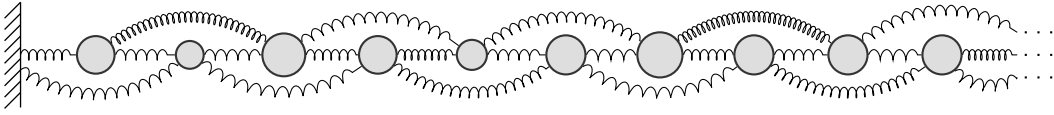
\begin{figure}[h]
\begin{center}
\begin{tikzpicture}
  [mass1/.style={circle,draw=black!80,fill=black!13,thick,inner sep=0pt,
   minimum size=5mm},
   mass2/.style={circle,draw=black!80,fill=black!13,thick,inner sep=0pt,
   minimum size=3.7mm},
   mass3/.style={circle,draw=black!80,fill=black!13,thick,inner sep=0pt,
   minimum size=5.7mm},
   mass4/.style={circle,draw=black!80,fill=black!13,thick,inner sep=0pt,
   minimum size=5mm}, 
   mass5/.style={circle,draw=black!80,fill=black!13,thick,inner sep=0pt,
   minimum size=4mm}, 
   mass6/.style={circle,draw=black!80,fill=black!13,thick,inner sep=0pt,
   minimum size=5.2mm}, 
   mass7/.style={circle,draw=black!80,fill=black!13,thick,inner sep=0pt,
   minimum size=6mm}, 
   mass8/.style={circle,draw=black!80,fill=black!13,thick,inner sep=0pt,
   minimum size=5.2mm},
   mass9/.style={circle,draw=black!80,fill=black!13,thick,inner sep=0pt,
   minimum size=5.4mm},
   massn/.style={circle,draw=black!80,fill=black!13,thick,inner sep=0pt,
   minimum size=5.2mm},
   wall/.style={postaction={draw,decorate,decoration={border,angle=-45,
   amplitude=0.3cm,segment length=1.5mm}}}
]
  \node (massn) at (12.75,1) [massn] {};
  \node (mass9) at (11.5,1) [mass8] {};
  \node (mass8) at (10.25,1) [mass8] {};
  \node (mass7) at (9.0,1) [mass7] {};
  \node (mass6) at (7.75,1) [mass6] {};
  \node (mass5) at (6.5,1) [mass5] {};
  \node (mass4) at (5.25,1) [mass4] {};
  \node (mass3) at (4.0,1) [mass3] {};
  \node (mass2) at (2.75,1) [mass2] {};
  \node (mass1) at (1.5,1) [mass1] {};
\draw[decorate,decoration={coil,aspect=0.4,segment
  length=1.1mm,amplitude=0.7mm}] (0.5,1) -- node[below=4pt]
{} (mass1);
\draw[decorate,decoration={coil,aspect=0.4,segment
  length=1.4mm,amplitude=0.7mm}] (mass1) -- node[below=4pt]
{} (mass2);
\draw[decorate,decoration={coil,aspect=0.4,segment
  length=1.5mm,amplitude=0.7mm}] (mass2) -- node[below=4pt]
{} (mass3);
\draw[decorate,decoration={coil,aspect=0.4,segment
  length=1.1mm,amplitude=0.7mm}] (mass3) -- node[below=4pt]
{} (mass4);
\draw[decorate,decoration={coil,aspect=0.4,segment
  length=0.9mm,amplitude=0.7mm}] (mass4) -- node[below=4pt]
{} (mass5);
\draw[decorate,decoration={coil,aspect=0.4,segment
  length=1.4mm,amplitude=0.7mm}] (mass5) -- node[below=4pt]
{} (mass6);
\draw[decorate,decoration={coil,aspect=0.4,segment
  length=1.7mm,amplitude=0.7mm}] (mass6) -- node[below=4pt]
{} (mass7);
\draw[decorate,decoration={coil,aspect=0.4,segment
  length=0.8mm,amplitude=0.7mm}] (massn) -- node[below=4pt]
{} (13.75,1);
\draw[decorate,decoration={coil,aspect=0.4,segment
  length=1.1mm,amplitude=0.7mm}] (mass7) -- node[below=4pt]
{} (mass8);
\draw[decorate,decoration={coil,aspect=0.4,segment
  length=1.3mm,amplitude=0.7mm}] (mass8) -- node[below=4pt]
{} (mass9);
\draw[decorate,decoration={coil,aspect=0.4,segment
  length=1.7mm,amplitude=0.7mm}] (mass9) -- node[below=4pt]
{} (massn);
\draw[decorate,decoration={coil,aspect=0.4,segment
  length=0.8mm,amplitude=0.7mm}] (mass1) to [bend left=35] node[below=4pt]
{} (mass3);
\draw[decorate,decoration={coil,aspect=0.4,segment
  length=1.5mm,amplitude=0.7mm}] (mass3) to [bend left=35] node[below=4pt]
{} (mass5);
\draw[decorate,decoration={coil,aspect=0.4,segment
  length=1.3mm,amplitude=0.7mm}] (mass5) to [bend left=35] node[below=4pt]
{} (mass7);
\draw[decorate,decoration={coil,aspect=0.4,segment
  length=0.7mm,amplitude=0.7mm}] (mass7) to [bend left=35] node[below=4pt]
{} (mass9);
\draw[decorate,decoration={coil,aspect=0.4,segment
  length=1.5mm,amplitude=0.7mm}] (mass9) to [bend left=35] node[below=4pt]
{} (13.75,1.3);
\draw[decorate,decoration={coil,aspect=0.4,segment
  length=1.5mm,amplitude=0.7mm}] (0.5,0.8) to [bend right=35] node[below=4pt]
{} (mass2);
\draw[decorate,decoration={coil,aspect=0.4,segment
  length=1.8mm,amplitude=0.7mm}] (mass2) to [bend right=35] node[below=4pt]
{} (mass4);
\draw[decorate,decoration={coil,aspect=0.4,segment
  length=1.2mm,amplitude=0.7mm}] (mass4) to [bend right=35] node[below=4pt]
{} (mass6);
\draw[decorate,decoration={coil,aspect=0.4,segment
  length=1.8mm,amplitude=0.7mm}] (mass6) to [bend right=35] node[below=4pt]
{} (mass8);
\draw[decorate,decoration={coil,aspect=0.4,segment
  length=1.1mm,amplitude=0.7mm}] (mass8) to [bend right=35] node[below=4pt]
{} (massn);
\draw[decorate,decoration={coil,aspect=0.4,segment
  length=1.5mm,amplitude=0.7mm}] (massn) to [bend right=35] node[below=4pt]
{} (13.75,0.7);
\draw[line width=.8pt,loosely dotted] (13.85,1) -- (14.25,1);
\draw[line width=.8pt,loosely dotted] (13.85,1.3) -- (14.25,1.3);
\draw[line width=.8pt,loosely dotted] (13.85,0.7) -- (14.25,0.7);
\draw[line width=.5pt,wall](0.5,1.7)--(0.5,0.3);
\end{tikzpicture}
\end{center}
\caption{Mass-spring system of a matrix in
  $\mathcal{M}(2,\infty)$: nondegenerated case}\label{fig:1}
\end{figure}
If for another matrix in $\mathcal{M}(2,\infty)$, one has degeneration
of the diagonals, for instance $m_1=4$, the corresponding mass-spring
system is given in Fig.~\ref{fig:2}.
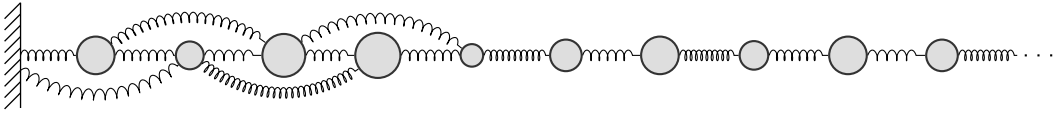
\begin{figure}[h]
\begin{center}
\begin{tikzpicture}
  [mass1/.style={circle,draw=black!80,fill=black!13,thick,inner sep=0pt,
   minimum size=5mm},
   mass2/.style={circle,draw=black!80,fill=black!13,thick,inner sep=0pt,
   minimum size=3.7mm},
   mass3/.style={circle,draw=black!80,fill=black!13,thick,inner sep=0pt,
   minimum size=5.7mm},
   mass4/.style={circle,draw=black!80,fill=black!13,thick,inner sep=0pt,
   minimum size=6mm}, 
   mass5/.style={circle,draw=black!80,fill=black!13,thick,inner sep=0pt,
   minimum size=3mm}, 
   mass6/.style={circle,draw=black!80,fill=black!13,thick,inner sep=0pt,
   minimum size=4.2mm}, 
   mass7/.style={circle,draw=black!80,fill=black!13,thick,inner sep=0pt,
   minimum size=5mm}, 
   mass8/.style={circle,draw=black!80,fill=black!13,thick,inner sep=0pt,
   minimum size=3.8mm}, 
   mass9/.style={circle,draw=black!80,fill=black!13,thick,inner sep=0pt,
   minimum size=5mm}, 
   massn/.style={circle,draw=black!80,fill=black!13,thick,inner sep=0pt,
   minimum size=4.2mm},
   wall/.style={postaction={draw,decorate,decoration={border,angle=-45,
   amplitude=0.3cm,segment length=1.5mm}}}
   ]
  \node (massn) at (12.75,1) [massn] {};
  \node (mass9) at (11.5,1) [mass9] {};
  \node (mass8) at (10.25,1) [mass8] {};
  \node (mass7) at (9.0,1) [mass7] {};
  \node (mass6) at (7.75,1) [mass6] {};
  \node (mass5) at (6.5,1) [mass5] {};
  \node (mass4) at (5.25,1) [mass4] {};
  \node (mass3) at (4.0,1) [mass3] {};
  \node (mass2) at (2.75,1) [mass2] {};
  \node (mass1) at (1.5,1) [mass1] {};
\draw[decorate,decoration={coil,aspect=0.4,segment
  length=1.0mm,amplitude=0.7mm}] (0.5,1) -- node[below=4pt]
{} (mass1);
\draw[decorate,decoration={coil,aspect=0.4,segment
  length=1.1mm,amplitude=0.7mm}] (mass1) -- node[below=4pt]
{} (mass2);
\draw[decorate,decoration={coil,aspect=0.4,segment
  length=1.3mm,amplitude=0.7mm}] (mass2) -- node[below=4pt]
{} (mass3);
\draw[decorate,decoration={coil,aspect=0.4,segment
  length=1.1mm,amplitude=0.7mm}] (mass3) -- node[below=4pt]
{} (mass4);
\draw[decorate,decoration={coil,aspect=0.4,segment
  length=1.1mm,amplitude=0.7mm}] (mass4) -- node[below=4pt]
{} (mass5);
\draw[decorate,decoration={coil,aspect=0.4,segment
  length=0.8mm,amplitude=0.7mm}] (mass5) -- node[below=4pt]
{} (mass6);
\draw[decorate,decoration={coil,aspect=0.4,segment
  length=1.1mm,amplitude=0.7mm}] (mass6) -- node[below=4pt]
{} (mass7);
\draw[decorate,decoration={coil,aspect=0.4,segment
  length=1.0mm,amplitude=0.7mm}] (mass8) -- node[below=4pt]
{} (mass9);
\draw[decorate,decoration={coil,aspect=0.4,segment
  length=0.8mm,amplitude=0.7mm}] (massn) -- node[below=4pt]
{} (13.75,1);
\draw[decorate,decoration={coil,aspect=0.4,segment
  length=0.7mm,amplitude=0.7mm}] (mass7) -- node[below=4pt]
{} (mass8);
\draw[decorate,decoration={coil,aspect=0.4,segment
  length=1.3mm,amplitude=0.7mm}] (mass9) -- node[below=4pt]
{} (massn);
\draw[decorate,decoration={coil,aspect=0.4,segment
  length=1.1mm,amplitude=0.7mm}] (mass1) to [bend left=35] node[below=4pt]
{} (mass3);
\draw[decorate,decoration={coil,aspect=0.4,segment
  length=1.3mm,amplitude=0.7mm}] (mass3) to [bend left=35] node[below=4pt]
{} (mass5);
\draw[decorate,decoration={coil,aspect=0.4,segment
  length=1.5mm,amplitude=0.7mm}] (0.5,0.8) to [bend right=35] node[below=4pt]
{} (mass2);
\draw[decorate,decoration={coil,aspect=0.4,segment
  length=0.8mm,amplitude=0.7mm}] (mass2) to [bend right=35] node[below=4pt]
{} (mass4);
 \draw[line width=.8pt,loosely dotted] (13.85,1) -- (14.25,1);
\draw[line width=.5pt,wall](0.5,1.7)--(0.5,0.3);
\end{tikzpicture}
\end{center}
\caption{Mass-spring system of a matrix in
  $\mathcal{M}(2,\infty)$: degenerated case}\label{fig:2}
\end{figure}

In this work, the approach to the inverse spectral analysis of the
operators whose matrix representation belongs to
$\mathcal{M}(n,\infty)$ is based on the one used in
\cite{MR1668981,MR1699440,2014arXiv1409.3868K} which deal with the
finite dimensional case. As in those papers, an important ingredient
of the inverse spectral analysis is the
linear interpolation of $n$-dimensional vector polynomials, recently
developed in \cite{2014arXiv1401.5384K}.

This paper is organized as follows. In Section \ref{sec:submatrices},
we present the results obtained in \cite{2014arXiv1409.3868K} on the
spectral measures of the operators corresponding to finite dimentional
matrices being an upper-left corner of a matrix in
$\mathcal{M}(n,\infty)$.  These finite dimensional operators play an
auxiliary role in the spectral analysis of operator $A$.  Later, in
Section~\ref{sec:general-case}, we construct a matrix valued function
for each element of $\mathcal{M}(n,\infty)$ having the properties of a
spectral function. Section \ref{sec:Spectral functions self-adjoint}
deals with various criteria for the operator $A$ to be self-adjoint
and gives the spectral function of $A$ touching uppon some of their
properties. Finally, in Section \ref{sec:Reconstruction}, we deal with
the problem of reconstruction and characterization.

\section{Spectral analysis of submatrices}
\label{sec:submatrices}

Fix $N>n$. The spectral analysis of the operator $A$ is carried out by
means of the auxiliary operator $P_{\cH_N}A\upharpoonright_{\cH_N}$,
where $\cH_N=\Span\{\delta_i\}_{i=1}^N$ and $P_{\cH_N}$ is the
orthogonal projection onto the subspace $\cH_N$.  Note that
$P_{\cH_N}A\upharpoonright_{\cH_N}$ can be identified with the
operator whose matrix representation is the finite dimensional
submatrix corresponding to the $N\times N$ upper-left corner of a
matrix in $\mathcal{M}(n,\infty)$ (cf. (\ref{eq:matrix-example})). We
denote the class of these $N \times N$ matrices by $\mathcal{M}(n,N)$
and the corresponding operator in $\cH_N$ is denoted by
$\widetilde{A}_N$.

According to \cite[Sec.\,2]{2014arXiv1409.3868K}, the spectral
analysis of the operator $\widetilde{A}_N$ can be carried out by
studying a system of $N$ equations, where each equation, given by a
fixed $k\in\{1,\dots,N\}$, is of the form
(cf. \cite[Eq.\,2.2]{2014arXiv1409.3868K})
\begin{equation}
\label{eq:recurrence}
\sum_{i=0}^{n-1} d_{k-n+i}^{(n-i)}\varphi_{k-n+i}+
d_k^{(0)}\varphi_k+
\sum_{i=1}^nd_k^{(i)}\varphi_{k+i}=z \varphi_k\,,
\end{equation}
where it has been assumed that
\begin{subequations}
  \label{eq:first-boundary-conditions-scalar}
\begin{align}
\label{eq:first-boundary-conditions-a-scalar}
\varphi_{k}=0\,,\quad\text{for}\ k&<1\,,\\
\label{eq:first-boundary-conditions-b-scalar}
\varphi_{k}=0\,,\quad\text{for}\ k&>N\,.
\end{align}
\end{subequations}

One can consider (\ref{eq:first-boundary-conditions-scalar}) as boundary
conditions where (\ref{eq:first-boundary-conditions-a-scalar}) is the condition
at the left endpoint and (\ref{eq:first-boundary-conditions-b-scalar}) is the
condition at the right endpoint.

The system (\ref{eq:recurrence}) with
(\ref{eq:first-boundary-conditions-scalar}), restricted to $k\in
\{1,2,\dots,N\}\setminus\{m_i\}_{i=1}^{j_0}$, can be solved
recursively whenever the first $n$ entries of the vector $\varphi$ are
given. Let $\varphi^{(j)}(z)$ ($j\in\{1,\dots,n\}$) be a solution of
(\ref{eq:recurrence}) for all $k\in
\{1,2,\dots,N\}\setminus\{m_i\}_{i=1}^{n}$ such that
\begin{equation}
  \label{eq:initial-condition}
\inner{\delta_i}{\varphi^{(j)}(z)}=t_{ji}, \;\text{for}\; i=1,\dots,n\,,
\end{equation}
where $\mathscr{T}=\{t_{ji}\}_{j,i=1}^n$ satisfies
\begin{enumerate}[I)]
\item $\mathscr{T}$ is $n \times n$ upper triangular with real
  entries.\label{item:prop-T}
\item $\prod_{i=1}^nt_{ii}\ne 0$.\label{item:prop-T-2}
\end{enumerate}

The condition given by (\ref{eq:initial-condition}) can be seen as the
initial conditions for the system (\ref{eq:recurrence}) and
(\ref{eq:first-boundary-conditions-a-scalar}). We emphasize that given
the boundary condition at the left endpoint
(\ref{eq:first-boundary-conditions-a-scalar}) and the initial
condition (\ref{eq:initial-condition}), the system restricted to
$k\in\{1,2,\dots,N\}\setminus\{m_i\}_{i=1}^{n}$ has a unique solution
for any fixed $j\in\{1,\dots,n\}$ and $z\in\mathbb{C}$. The
degenerations, which the diagonals of matrices in $\mathcal{M}(n,N)$
undergo, are related to other kind of ``boundary conditions''. Indeed,
the equations of the system (\ref{eq:recurrence}), when
$k\in\{m_j\}_{j=1}^{j_0}$, give rise to the inner boundary conditions
(of the right endpoint type)
(cf. \cite[Eq.\,2.8]{2014arXiv1409.3868K}).

The normalized eigenvectors of the
operator $\widetilde{A}_N$ can be decomposed as follows
\begin{equation}
\label{eq:eigenvector}
  \alpha(x_l)=\sum_{j=1}^n\alpha_j(x_l)\varphi^{(j)}(x_l)\,,
\end{equation}
where $\{x_l\}_{l=1}^N=:\spec\widetilde{A}_N$ and
$\alpha_j(x_l)\in\complex$. It follows from
(\ref{eq:recurrence}), (\ref{eq:first-boundary-conditions-scalar}), and
(\ref{eq:initial-condition}), that
\begin{equation*}
\sum_{j=1}^n\abs{\alpha_j(x_k)}>0\quad\text{for all }\,
k\in\{1,\dots,N\}
\end{equation*}
and
\begin{equation*}
\sum_{k=1}^N\abs{\alpha_j(x_k)}>0\quad\text{for all }\,
j\in\{1,\dots,n\}\,.
\end{equation*}

The operator $\widetilde{A}_N$ has a matrix-valued spectral function
\begin{align}
  \label{eq:spec-func-form}
  \sigma_N^{\mathscr{T}}(t)=&\sum_{x_l<t}\begin{pmatrix}
    \abs{\alpha_1(x_l)}^2&\overline{\alpha_1(x_l)}\alpha_2(x_l)&
\dots&\overline{\alpha_1(x_l)}\alpha_n(x_l)\\
    \overline{\alpha_2(x_l)}\alpha_1(x_l)&\abs{\alpha_2(x_l)}^2&
\dots&\overline{\alpha_2(x_l)}\alpha_n(x_l)\\
    \vdots&\vdots&\ddots&\vdots\\
    \overline{\alpha_n(x_l)}\alpha_1(x_l)&\overline{\alpha_n(x_l)}
\alpha_2(x_l)&\dots&\abs{\alpha_n(x_l)}^2
\end{pmatrix}
\end{align}
with the following properties:
\begin{enumerate}[a)]
\item It is a nondecreasing monotone step function which is continuous
  from the left.
\item Each jump is a matrix of rank not greater than $n$.
\item The sum of the ranks of all jumps equals $N$.
\end{enumerate}

Note that the matrices in the sum on the r.\,h.\,s. of
(\ref{eq:spec-func-form}) are the tensor product of
the vector $\left(\begin{smallmatrix}\overline{\alpha_1(x_l)}\\
    \vdots\\\overline{\alpha_n(x_l)}\end{smallmatrix}\right)$ with the
complex conjugate of itself.

The relationship between the spectral functions
$\sigma_N^{\mathscr{T}}$ for an arbitrary $\mathscr{T}$ and the case
$\mathscr{T}=I$ is given by the following equation which is proven in
\cite[Pro.\,2.1]{2014arXiv1409.3868K}.
\begin{equation}
\label{eq:finite-measure-finite-first-moment}
  \mathscr{T}^*\int_{\reals}d\,\sigma_N^{\mathscr{T}}\mathscr{T}
=\int_{\reals}d\,\sigma_N^I=I\,.
\end{equation}

Consider the Hilbert space $L_2(\mathbb{R},\sigma_N^{\mathscr{T}})$
with the usual inner product which we assume to be antilinear in the
first argument (for the definition of
$L_2(\mathbb{R},\sigma_N^{\mathscr{T}})$,
see \cite[Sec.\,72]{MR1255973}).  Clearly, the property c) implies
that $L_2(\mathbb{R},\sigma_N^{\mathscr{T}})$
is an $N$-dimensional
space and in each equivalence class there is an $n$-dimensional
vector polynomial.

Define the vectors
\begin{equation}
\label{eq:first-pk}
\pb{p}_k:=\mathscr{T}\pb{e}_k\quad\text{ for }k=1,\dots,n\,,
\end{equation}
where $\{\pb{e}_k\}_{k=1}^n$ is the canonical basis in
$\complex^n$, i.e.,
\begin{equation}
  \label{eq:canonical-basis-CN}
  \pb{e}_1=\left(
    \begin{tiny}
      \begin{matrix}
        1\\
        0\\
        \vdots\\
        0
      \end{matrix}
    \end{tiny}
\right),\pb{e}_2=\left(
    \begin{tiny}
      \begin{matrix}
        0\\
        1\\
        \vdots\\
        0
      \end{matrix}
    \end{tiny}
\right),\dots,\pb{e}_n=\left(
    \begin{tiny}
      \begin{matrix}
        0\\
        \vdots\\
        0\\
        1
      \end{matrix}
    \end{tiny}
\right)\,.
\end{equation}
Taking $\{\pb{p}_k\}_{k=1}^n$ as initial conditions of
the recurrence equation
\begin{equation}
\label{eq:recurrence-vector-plynomials}
\sum_{i=0}^{n-1} d_{k-n+i}^{(n-i)}\pb{p}_{k-n+i}(z)+
d_k^{(0)}\pb{p}_k(z)+
\sum_{i=1}^nd_k^{(i)}\pb{p}_{k+i}(z)=z \pb{p}_k(z)\,,\quad
k\in\nats\setminus\{m_j\}_{j=1}^{j_0}\,,
\end{equation}
where it is assumed that
\begin{align}
\label{eq:first-boundary-conditions}
\pb{p}_{k}=0\,,\quad\text{for}\quad k<1\,,
\end{align}
one obtains a sequence $\{\pb{p}_k(z)\}_{k=1}^\infty$ of vector
polynomials. The next assertion is proven in
\cite[Lem.\,2.2]{2014arXiv1409.3868K}.
\begin{proposition}
  \label{prop:ortonormal-p-L2-finite}
  For any natural number $N>n$, the vector polynomials
  $\{\boldsymbol{p}_k(z)\}_{k=1}^{N}$, defined by
  (\ref{eq:recurrence-vector-plynomials}), satisfy
\begin{equation*}
  \inner{\boldsymbol{p}_j}{\boldsymbol{p}_k}
_{L_2(\mathbb{R},\sigma_N^{\mathscr{T}})}
=\delta_{jk}
\end{equation*}
for $j,k\in\{1,\dots,N\}$.
\end{proposition}

Let $U:\mathcal{H_N} \rightarrow L_2(\mathbb{R},\sigma_N^{\mathscr{T}})$
be the isometry given by $U\delta_k\mapsto \boldsymbol{p}_k$, for all
$k\in\{1,\dots,N\}$. Under this isometry the operator
$\widetilde{A}_N$ becomes the operator of multiplication by the
independent variable in $L_2(\mathbb{R},\sigma_N^{\mathscr{T}})$ (see
\cite[Sec.~2]{2014arXiv1409.3868K}).

Define
\begin{equation}
  \label{eq:polynomials_q} \pb{q}_j(z):=
(z-d_{m_j}^{(0)})\pb{p}_{m_j}(z)-\sum_{k=0}^{n-1}
d_{m_j-n+k}^{(n-k)}\pb{p}_{m_j-n+k}(z)-\sum_{k=1}^{n-j}d_{m_j}^{(k)}\pb{p}_{m_j+k}(z)
\end{equation}
for $j\in\{1,\dots,j_0\}$.

Using the same reasoning as in \cite[Thm.\,3.1]{2014arXiv1409.3868K},
one proves that, for any natural number $N\ge n_0+m_{j_0}$ (see
  Remark~\ref{rem:tail-matrix}), the vector polynomials
  $\{\boldsymbol{q}_j(z)\}_{k=1}^{j_0}$ satisfy
\begin{equation}
\label{eq:q-are-zeros}
  \inner{\boldsymbol{q}_j}{\boldsymbol{q}_j}
_{L_2(\mathbb{R},\sigma_N^{\mathscr{T}})}
=0\,.
\end{equation}

The existence of polynomials of zero norm in
$L_2(\mathbb{R},\sigma_N^{\mathscr{T}})$ is related to a linear
interpolation problem consisting in the following: Given collections
of numbers $\{z_k\}_{k=1}^N$ and $\{\alpha_j(k)\}_{j=1}^n$
($k=1,\dots,N$), find the scalar polynomials $R_{j}(z)$ (j=1,\dots,n)
wich satisfy the equation
\begin{equation*}
  \sum_{j=1}^n\alpha_{j}(k)R_{j}(z_k)=0\,,\quad \forall k\in\{1,\dots,N\}\,.
\end{equation*}
This is equivalent (see \cite{2014arXiv1409.3868K}) to
finding  $n$-dimensional vector polynomials satisfying
\begin{equation}
  \label{eq:interpolation-vector-polynomials}
  \inner{\boldsymbol{r}(z)}{\boldsymbol{r}(z)}
_{L_2(\reals,\sigma_N^{\mathscr{T}})}=0\,,
  \qquad \boldsymbol{r}(z)=\left(R_1(z),R_2(z),\ldots,R_n(z)\right)^{t}\,.
\end{equation}
In \cite{2014arXiv1401.5384K} it was found that the solutions of the
linear interpolation problem given by
(\ref{eq:interpolation-vector-polynomials}) are determined by a set of
$n$ vector polynomials called generators
\cite[Thm.\,5.3]{2014arXiv1401.5384K}.

An important concept in the context of solving 
\eqref{eq:interpolation-vector-polynomials} is the following.
 \begin{definition}
\label{def:height}
Let $\boldsymbol{r}(z)=\left(R_1(z),R_2(z),\ldots,R_n(z)\right)^{t}$
be an $n$-dimensional vector polynomial. The height of
$\boldsymbol{r}(z)$ is the number
\begin{equation*}
h(\boldsymbol{r}):=
\max_{j\in\{1,\dots,n\}}\left\lbrace n\deg (R_j)+j-1\right\rbrace\,,
\end{equation*}
where it is assumed that $\deg 0:=-\infty$ and
$h(\boldsymbol{0}):=-\infty$.
\end{definition}

Note that we have defined the vector polynomials
$\{\pb{e}_k\}_{k=1}^n$ so that
\begin{equation}
  \label{eq:height-e}
  h(\pb{e}_k)=k-1\,.
\end{equation}
Having the concepts of height of a vector polynomial and generator of
the interpolation problem \eqref{eq:interpolation-vector-polynomials}
at hand, we invoke results from \cite{2014arXiv1409.3868K} and
\cite{2014arXiv1401.5384K}. First we convene:
\begin{convention}
  From now on, we consider the natural number $N$ to be no less than
  $n_0+m_{j_0}$ (see Remark~\ref{rem:tail-matrix}).
\end{convention}

\begin{proposition}
\label{prop:q-j-generator}
(\cite[Thm.\,3.1]{2014arXiv1409.3868K}) The vector polynomials
$\{\boldsymbol{q}_j(z)\}_{j=1}^{j_0}$ are the first $j_0$ generators
of the linear interpolation problem given by
(\ref{eq:interpolation-vector-polynomials}) (see
\cite[Sec.\,3]{2014arXiv1409.3868K}).  Moreover, for $j=1,\dots,j_0$,
$h(\pb{q}_j)$ are different elements of the factor space
$\integers/n\integers$ \cite[Lem.\,4.3]{2014arXiv1401.5384K}.
\end{proposition}

The heights of the vector polynomials $\{\pb{p}_k\}_{k=n+1}^\infty$
are determined recursively by means of the system
\eqref{eq:recurrence-vector-plynomials}. Indeed, for any
$m_{j}<k<m_{j+1}$, with $j=0,\dots,j_0$, one has the equation
  \begin{equation*}
    \dots +d_k^{(0)}\pb{p}_k+d_k^{(1)}\pb{p}_{k+1}+\dots+
d_{k}^{(n-j)}\pb{p}_{k+n-j}=z\pb{p}_k\,,
  \end{equation*}
  where we have assumed that $m_0=0$. Since $d_{k}^{(n-j)}$ never
  vanishes, the height of $\pb{p}_{k+n-j}$ coincides with the one of
  $z\pb{p}_k$. Thus
  \begin{equation}
    \label{eq:height-p-k}
    h(\pb{p}_{k+n-j})=n+h(\pb{p}_{k})\,.
  \end{equation}
  If there are no degenerations of the diagonals, then
  (\ref{eq:height-p-k}) implies that
  \begin{equation}
\label{eq:height-p-non-degeneration}
    h(\pb{p}_{k})=k-1\,,\quad \text{for all }k\in\nats\,.
  \end{equation}
  On the other hand, in the presence of degenerations, one verifies
  from  (\ref{eq:polynomials_q}) and (\ref{eq:height-p-k}) that, no
  matter which $k\in\nats$ one chooses,
\begin{equation}
\label{eq:height-different-p-and-q}
  h(\pb{p}_{k})\ne h(\pb{p}_{m_j})+n=h(\pb{q}_{j})\,,
\end{equation}
for any $j=1,\dots,j_0$.

\begin{lemma}
  \label{lem:cover-all-heights}
  For any nonnegative integer $s$, there exist $k\in\nats$ or a pair
  $j\in\{1,\dots,j_0\}$ and $l\in\nats\cup\{0\}$ such that either
  $s=h(\pb{p}_{k})$ or $s=h(\pb{q}_{j})+nl$.
\end{lemma}
\begin{proof}
  This proof repeats the one of
  \cite[Lem.\,3.3]{2014arXiv1409.3868K}. We have reproduced it here
  for the reader's convenience. Due to (\ref{eq:height-p-k}), it
  follows from \eqref{eq:first-pk} and \eqref{eq:height-e} that
  \begin{equation}
\label{eq:first-p-height}
    h(\pb{p}_{k})=k-1\quad\text{ for }k=1,\dots,h(\pb{q}_1)
\end{equation}
(cf. \eqref{eq:height-p-non-degeneration}).

Suppose that there is $s\in\nats$ ($s>n$) such that $s\ne h(\pb{p}_k)$
for all $k\in\nats$ and $s\ne h(\pb{q}_j)+nl$ for all
$j\in\{1,\dots,j_0\}$ and $l\in\nats\cup\{0\}$. Let $\hat{l}$ be an
integer such that
$s-n\hat{l}\in\{h(\pb{p}_k)\}_{k=1}^\infty\cup\{h(\pb{q}_j)+nl\}$
($j\in\{1,\dots,j_0\}$ and $l\in\nats\cup\{0\}$). There is always such
an integer due to (\ref{eq:first-p-height}) and $h(\pb{q}_1)>n$ (see
(\ref{eq:height-different-p-and-q})). We take $\hat{l}_0$ to be the
minimum of all $\hat{l}$'s. Thus, there is $k'\in\nats$ or
$j'\in\{1,\dots,j_0\}$, respectively, such that either
\begin{enumerate}[a)]
\item $s-n\hat{l}_0=h(\pb{p}_{k'})$ or \label{item:lem:heigth-nats-1}
 \item  $s-n\hat{l}_0=h(\pb{q}_{j'})+nl$, with
$l\in\nats\cup\{0\}$.\label{item:lem:heigth-nats-2}
\end{enumerate}
In the case \ref{item:lem:heigth-nats-1}), we prove that $\hat{l}_0$
is not the minimum integer, this implies the assertion of the
lemma. Indeed, if there is $j\in \{1,\dots,j_0\}$ such that $k'=m_{j}$,
then $s-n\hat{l}_0+n=h(\pb{p}_{m_{j}})+n=h(\pb{q}_{j})$ due to
(\ref{eq:height-different-p-and-q}). If there is not such $j$, then
$m_j<k'<m_{j+1}$ and (\ref{eq:height-p-k}) implies
$s-n\hat{l}_0+n=h(\pb{p}_{k'})+n=h(\pb{p}_{k'+n-j})$.

For the case \ref{item:lem:heigth-nats-2}), if
$s-n\hat{l}_0=h(\pb{q}_{j'})+nl$, then
$s=h(\pb{q}_{j'})+n(l+\hat{l}_0)$ which is a contradiction.
\end{proof}
As a consequence of \cite[Thm.~2.1]{2014arXiv1401.5384K}, the above
lemma yields the following result.
\begin{corollary}
\label{cor:basis-vector}
Any vector polynomial $\boldsymbol{r}(z)$ is a finite linear
combination of
\begin{equation}
\label{eq:basis-p-q}
  \{\pb{p}_k(z): k\in\nats\} \cup \{z^l\pb{q}_j(z):
l\in\nats\cup\{0\},\, j\in\{1,\dots,j_0\}\}\,,
\end{equation}
where if $j_0=0$, the second set in \eqref{eq:basis-p-q} is empty.
\end{corollary}

To conclude this section, we use the canonical basis of $\complex^n$
(see (\ref{eq:canonical-basis-CN})) to define the family of vector
polynomials for $k\in\nats$ and $i=1,\dots,n$.
\begin{equation}
    \label{eq:canonical-vector-polynomial}
    \pb{e}_{nk+i}(z):=z^k\pb{e}_i\,.
  \end{equation}
Observe that
\begin{equation}
\label{eq:remark-basis-canonical-moments}
\inner{\pb{e}_{nk+i}(t)}{\pb{e}_{nl+j}(t)}_{L_2(\reals,\sigma_N^{\mathscr{T}})}=
\int_{\reals}t^{k+l}d\,\sigma_{N}^{\mathscr{T}}(i,j)\,.
\end{equation}
On the basis of Corollary~\ref{cor:basis-vector}, one verifies that
the matrix moments of $\sigma_N^{\mathscr{T}}$,
\begin{equation}
\label{eq:finite-moments}
S_k(\mathscr{T}):=\int_{\reals}t^k\,d\sigma_N^{\mathscr{T}}\quad \text{for}\ 
k=0,1,\dots,\left\lceil\frac{2h(\pb{p}_N)}{n}\right\rceil\,,
\end{equation}
coincide with the ones of $\sigma_{\widetilde{N}}^{\mathscr{T}}$ for any $\widetilde{N}\geq N$, where $\lceil\cdot\rceil$ is
the ceiling function.
\begin{remark}
  \label{rem:no-degeneration-moments}
  Note that, for any natural number $k$, there exists $N\in\nats$ such
  that $S_{2k}(\mathscr{T})$, given by (\ref{eq:finite-moments}), is a
  positive definite matrix.
\end{remark}

\section{Spectral analysis of infinite symmetric band matrices}
\label{sec:general-case}

In this section, we construct a matrix valued function for each
element of $\mathcal{M}(n,\infty)$ having the properties of a spectral
function. To this end, we give defining criteria for a measure to be a
spectral function of a \emph{matrix} in the class
$\mathcal{M}(n,\infty)$. By our definition, any spectral function
$\sigma$ of $\mathcal{A}$ in $\mathcal{M}(n,\infty)$ is the spectral
function of some self-adjoint extension of the minimal closed operator
generated by $\mathcal{A}$ (see \cite[Sec.\,47]{MR1255973}) so that
this self-adjoint operator is transformed by a unitary isometric map,
which can be regarded as a Fourier transform, into the operator of
multiplication by the independent variable defined on its maximal
domain in some space $L_2(\reals,\sigma)$ (for the definition of this
space, see \cite[Sec.\,72]{MR1255973}). It is worth remarking that
not all the spectral functions of a matrix in $\mathcal{M}(n,\infty)$
correspond to a self-adjoint extension $A_0$ of the minimal closed
operator generated by $A$ such that $A_0\subset A^*$ (see
Remark~\ref{rem:s-a-extensions-measures}).

The results of this section and the next one provides a complete
description of all possible spectral functions that can be associated
with some element of $\mathcal{M}(n,\infty)$ by our criteria.


\begin{definition}
  \label{def:spectral-measure-gen-case}
  A nondecreasing $n\times n$
  matrix-valued function $\sigma$
  with finite moments, such that $\int_\reals d\sigma$
  is invertible, is called a spectral function of a matrix
  $\mathcal{A}$
  in $\mathcal{M}(n,\infty)$
  if and only if there exist $\mathscr{T}$
  satisfying \ref{item:prop-T}) and \ref{item:prop-T-2}) (see
  \eqref{eq:initial-condition}) such that
  $\{\pb{p}_{k}\}_{k=1}^\infty$
  is an orthonormal sequence in $L_2(\reals,\sigma)$
  and, for each $j\in\{1,\dots,j_0\}$,
  $\pb{q}_j$
  is in the equivalence class of zero in $L_2(\reals,\sigma)$.
\end{definition}

Note that all the vector polynomials are in $L_2(\reals,\sigma)$ when
$\sigma$ is a spectral function of a matrix in
$\mathcal{M}(n,\infty)$. Moreover the polynomials are dense in
$L_2(\reals,\sigma)$ when the orthonormal system
$\{\pb{p}_k\}_{k=1}^\infty$ turns out to be complete.

On the basis of the definition above, one can construct an isometric
map between the original space $\cH$ and the subspace being the
closure of the polynomials in $L_2(\reals,\sigma)$. This isometric
map, which will be denoted by $U$, is realized by associating the
orthonormal basis $\{\delta_k\}_{k=1}^\infty$ with the orthonormal
system $\{\pb{p}_k\}_{k=1}^\infty$, i.\,e., $U\delta_k=\pb{p}_k$ for
all $k\in\nats$. Furthermore, under this map, the operator $A$ is
transformed into some restriction of the operator of multiplication by
the independent variable. Indeed, if
$\varphi=\sum_{k=1}^\infty\varphi_k\delta_k$ is an element of the
domain of $A$, then $f=\sum_{k=1}^\infty\varphi_k\pb{p}_k$ is in the
domain of the operator of multiplication by the independent variable
and
\begin{equation*}
  UAU^{-1}f(t)=tf(t)\,.
\end{equation*}
\begin{lemma}
  \label{lem:symmetric-case}
  Let $\mathcal{A}$ be an element of $\mathcal{M}(n,\infty)$ and
  $\sigma_N^{\mathscr{T}}$ be the matrix valued spectral function of
  the corresponding oparator $A_N$ for a fixed matrix $\mathscr{T}$
  satisfying \ref{item:prop-T}) and \ref{item:prop-T-2}). Then, there
  exist a subsequence $\{\sigma_{N_i}^{\mathscr{T}}\}_{i=1}^{\infty}$
  converging pointwise to a matrix valued function $\sigma^{\mathscr{T}}$.
\end{lemma}
\begin{proof}
  In view of (\ref{eq:finite-measure-finite-first-moment}),
  the hypothesis of Helly's first theorem for bounded operators
  \cite[Thm.\,4.3]{MR0322542} is satisfied in any bounded interval
  (cf. \cite[Sec.\,8.4]{MR0067952} for the scalar case), therefore the
  statement follows. The generalization of Helly's first theorem given in
  \cite[Thm.\,4.3]{MR0322542} is based on applying the scalar theorem
  to the bilinear form of the sequence of operators (for fixed
  elements in the Hilbert space) in a diagonal process fashion using
  the boundedness of the operators and the seperability of the
  space. This yields the assertion in the sense of weak convergence.
  Using the fact that uniform and weak convergence are
  equivalent in finite dimentional spaces, one obtains the assertion.
\end{proof}


The following proposition is obtained by applying
\cite[Thm.\,4.4]{MR0322542} to the result above and taking into
account that the matrices $\sigma_N^{\mathscr{T}}$ are finite dimensional.

 \begin{proposition}
   \label{prop:2nd-helly}
   (Helly's generalized second theorem) Suppose that the function
   $f(t)$ is continuous in the real interval $[a,b]$, where $a$ and
   $b$ are points of continuity of $\sigma^{\mathscr{T}}(t)$ (see
   Lemma~\ref{lem:symmetric-case}). Then there exist a subsequence
  $\{\sigma_{N_i}^{\mathscr{T}}\}_{i=1}^{\infty}$ such that
   \begin{equation*}
     \int_{a}^bf(t)\,d\sigma^{\mathscr{T}}_{N_i}(t)\xrightarrow[i\rightarrow\infty]{}
\int_{a}^bf(t)\,d\sigma^{\mathscr{T}}(t)\,.
   \end{equation*}
 \end{proposition}

With these results at hand we prove the following assertions.
\begin{lemma}
  \label{lem:subsequence-converges-integral}
There exist a subsequence
  $\{\sigma_{N_i}^{\mathscr{T}}\}_{i=1}^{\infty}$ such that
\begin{equation*}
    \int_{\reals}t^k\,d\sigma^{\mathscr{T}}_{N_i}=\int_{\reals}t^k \,d\sigma^{\mathscr{T}}\
  \end{equation*}
 for any nonnegative integer $k\le \left\lceil
  \frac{2h(\pb{p}_{N_i})}{n}\right\rceil$ (see \eqref{eq:finite-moments}).
\end{lemma}
 \begin{proof}
  If one assumes that
  $-a<0$ and $b>0$ are two points of continuity of
  $\sigma^{\mathscr{T}}(t)$, then, it
  follows from Proposition~\ref{prop:2nd-helly} that
\begin{equation*}
   \int_{-a}^bt^k\,d\sigma^{\mathscr{T}}=\lim_{i\rightarrow\infty} 
\int_{-a}^bt^k\,d\sigma^{\mathscr{T}}_{N_i}\,.
\end{equation*}
On the other hand, given a number $r$ such that $r>k$, then for
$r<\left\lceil \frac{2h(\pb{p}_{N_i})}{n}\right\rceil$
\begin{align*}
  \norm{\int_{-\infty}^{\infty}-
    \int_{-a}^{b}t^k\,d\sigma^{\mathscr{T}}_{N_i}}=&\norm{\int_{-\infty}^{-a}+
    \int_{b}^{\infty}t^k\,d\sigma^{\mathscr{T}}_{N_i}}=
  \norm{\int_{-\infty}^{-a}+ \int_{b}^{\infty}\frac{t^r}{t^{r-k}}\,
      d\sigma^{\mathscr{T}}_{N_i}}\\
  \leq&\frac{1}{c^{r-k}}\norm{\int_{-\infty}^{-a}+
    \int_{b}^{\infty}t^r\,d\sigma^{\mathscr{T}}_{N_i}}\leq \frac{\norm{S_r(\mathscr{T})}}{c^{r-k}}\,,
\end{align*}
where $c=\min\{a,b\}$ and $S_r(\mathscr{T})=\int_{\reals}t^r\,d\sigma^{\mathscr{T}}_{N_i}$ (the
integral is convergent due to
Proposition~\ref{prop:ortonormal-p-L2-finite}). Thus,
\begin{equation*}
   \norm{S_k(\mathscr{T})-\int_{-a}^bt^k\,d\sigma^{\mathscr{T}}}\leq \frac{\norm{S_r(\mathscr{T})}}{c^{r-k}}\,.
\end{equation*}
This yields the assertion, when one makes $a$ and $b$
tend to $\infty$ in such a way that $-a$ and $b$ are all the time
points of continuity of $\sigma^{\mathscr{T}}(t)$.
\end{proof}
From the previous lemma, one directly obtains the following result.
\begin{corollary}
  The spectral function $\sigma^{\mathscr{T}}$ to which a subsequence
  of $\{\sigma_N^{\mathscr{T}}\}_{N=2}^\infty$ converges according to
  Lemma~\ref{lem:symmetric-case} is a solution of a certain
  matrix moment problem given by $\{S_k(\mathscr{T})\}_{k=0}^\infty$
  (see Lemma~\ref{lem:subsequence-converges-integral} ).
\end{corollary}

\begin{lemma}
  \label{lem:has-a-spectral-function}
  Any $\mathcal{A}\in\mathcal{M}(n,\infty)$ has at least one spectral
  function (in the sense of
  Definition~\ref{def:spectral-measure-gen-case}).
\end{lemma}
\begin{proof}
  It follows directly from
  Proposition~\ref{prop:ortonormal-p-L2-finite} and
  Lemma~\ref{lem:subsequence-converges-integral} that the vector
  polynomials $\{\boldsymbol{p}_k(z)\}_{k=1}^{\infty}$, defined by
  \eqref{eq:canonical-basis-CN} and
  (\ref{eq:recurrence-vector-plynomials}), satisfy
\begin{equation*}
  \inner{\boldsymbol{p}_j}{\boldsymbol{p}_k}
_{L_2(\mathbb{R},\sigma^{\mathscr{T}})}
=\delta_{jk}
\end{equation*}
for $j,k\in\nats$, where $\sigma^{\mathscr{T}}$ is the function given
by Lemma~\ref{lem:symmetric-case}.  Now, fix
$j\in\{1,\dots,j_0\}$ and consider $N$ according to Convention~1. Thus,
  \begin{equation*}
    0=\norm{\pb{q}_j}_{L_2(\reals,\sigma_N^{\mathscr{T}})}^2=
    \int_\reals\inner{\pb{q}_j}{d\sigma_N^{\mathscr{T}}\pb{q}_j}\,.
  \end{equation*}
By Lemma~\ref{lem:subsequence-converges-integral} there is a
subsequence  $\{\sigma_{N_i}^{\mathscr{T}}\}_{i=1}^{\infty}$ such
that, begining from some $i\in\nats$,
\begin{equation*}
 0= \int_\reals\inner{\pb{q}_j}{d\sigma_{N_i}^{\mathscr{T}}\pb{q}_j}=
\int_\reals\inner{\pb{q}_j}{d\sigma^{\mathscr{T}}\pb{q}_j}=
\norm{\pb{q}_j}_{L_2(\reals,\sigma^{\mathscr{T}})}^2
\end{equation*}
\end{proof}

\begin{corollary}
  \label{cor:measure-supp}
  The spectral function $\sigma^{\mathscr{T}}$ given in
  Lemma~\ref{lem:symmetric-case} has an infinite number of
  growth points.
\end{corollary}
\begin{proof}
  If $\sigma^{\mathscr{T}}$ had a finite number of
  growth points, then
  $L_2(\reals,\sigma^{\mathscr{T}})$ would be a finite dimensional space
  and correspondingly the sequence of vector polynomials $\{\pb{p}_k\}_k$
  would be finite.
\end{proof}

\begin{remark}
  \label{rem:s-a-extensions-measures}
  Let $\mathcal{A}$
  be in $\mathcal{M}(n,\infty)$
  and $\sigma$
  be the spectral function of $\mathcal{A}$
  according to Definition~\ref{def:spectral-measure-gen-case}. If the
  moment problem associated with $\sigma$
  turns out to be \emph{determinate}, then there is just one solution
  of the moment problem and this solution, i.\,e. $\sigma$,
  corresponds to a spectral function of the operator $A$
  which turns out to be self-adjoint \cite[Sec.\,2]{MR1882637}. Note
  that one could associate another spectral function to $A$
  by considering a different matrix $\mathscr{T}$
  (see \eqref{eq:initial-condition}), but the moment problem for it
  would be different. If the moment problem is \emph{indeterminate},
  then there are various solutions of the moment problem and each
  solution $\widehat{\sigma}$
  is a spectral function of $\mathcal{A}$
  since the sequence of polynomials $\{\pb{p}_k\}_{k=1}^\infty$
  is orthonormal in $L_2(\reals,\widehat{\sigma})$
  for any solution $\widehat{\sigma}$.
  In this case, $\widehat{\sigma}$
  not necessarily corresponds to the spectral function of canonical
  self-adjoint extensions of the operator $A$
  (by a canonical self-adjoint extension of the symmetric operator $A$
  we mean a self-adjoint restriction of $A^*$).
  Indeed, the solution $\widehat{\sigma}$
  is the spectral function of a canonical self-adjoint extension if
  and only if the polynomials are dense in
  $L_2(\reals,\widehat{\sigma})$.
  We expect that the spectral function $\sigma^{\mathscr{T}}$,
  to which a subsequence of $\{\sigma_N^{\mathscr{T}}\}_{N=2}^\infty$
  converges according to Lemma~\ref{lem:symmetric-case}, be such that
  the polynomials are dense in
  $L_2(\reals,\sigma^{\mathscr{T}})$.
  This matter, together with other questions on characterization of
  the functions $\sigma^{\mathscr{T}}$
  will be dealt with in a forthcoming paper.
\end{remark}

\begin{definition}
  \label{def:class-sigma-infinite}
  The set of all $n\times n$-matrix valued functions with an infinite
  number of growing points such that all the
  moments $\{S_k\}_{k=1}^\infty$ exists and $S_0$ is invertible is
  denoted by $\mathfrak{M}(n,\infty)$. Besides,
  $\mathfrak{M}_d(n,\infty)$ denotes the subset of
  $\mathfrak{M}(n,\infty)$ for which the sequence of matrix moments
  generates a determinate matrix moment problem.
\end{definition}

\begin{theorem}
  \label{thm:non-degenerate-case}
  Let $\mathcal{A}$ be in $\mathcal{M}(n,\infty)$ and $j_0$ be the number of
  degenerations of $\mathcal{A}$ (see the paragraph below
  Definition~\ref{def:matrices-degenerate}). For any spectral function
  $\sigma$ of $\mathcal{A}$, it holds true that:
  \begin{enumerate}[i)]
  \item \label{thm:non-case-degenarate} (Nondegenerate case) If
    $j_0=0$, i.\,e., the matrix $\mathcal{A}$ does not
    undergo degeneration, then there are no vector polynomials
    in the equivalence class of the zero of the space
    $L_2(\reals,\sigma)$, i.\,e.,
    \begin{equation*}
      \inner{\boldsymbol{r}(z)}{\boldsymbol{r}(z)}_{L_2(\reals,\sigma)}=0\,
      \Longleftrightarrow \boldsymbol{r}\equiv 0\,.
    \end{equation*}
  \item \label{thm:case-degenarate} (Degenerate case)
  If $j_0>0$, then all the polynomials $\pb{q}_1,\dots,\pb{q}_{j_0}$
  are  in the
    equivalence class of zero and any polynomial $\pb{r}(z)$ in this
    equivalence class can be written as
    \begin{equation}
      \label{thm:eq-case-degenarate}
      \pb{r}(z)=\sum_{j=1}^{j_0}R_j(z)\pb{q}_j(z)\,,
    \end{equation}
where $R_j(z)$ is a scalar polynomial.
  \end{enumerate}
\end{theorem}
\begin{proof}
  First one proves (\ref{thm:case-degenarate}). The first part of the
  assertion follows immediatly from Definition~\ref{def:spectral-measure-gen-case}.
 Suposse that there is
  a nontrivial vector polynomial $\pb{r}(z)$ in the equivalence class
  of zero with height $r$. Therefore, by
  Corollary~\ref{cor:basis-vector}
  \begin{equation}
\label{eq:expand-any-zero}
  \pb{r}(z)=\sum_{k=1}^{l} c_k\pb{p}_k(z)+\sum_{j=1}^{j_0}R_j(z)\pb{q}_j(z)\,,
\end{equation}
where $\max\{
h(\pb{p}_l),\max\limits_{j=1,\dots,j_0}\{h(R_j\pb{q}_j)\}\}=r$. Furthermore,
\begin{equation}
\label{eq:c-k-expand-any-zero}
  c_k=\inner{\pb{r}(z)}{\pb{p}_k(z)}_{L_2(\reals,\sigma)}\, 
\quad \text{ for all } k\in\{1,\dots,l\}\,. 
\end{equation}
And, since $\pb{r}(z)$ is in the zero class, the r.h.s. of the
equality in (\ref{eq:c-k-expand-any-zero}) is always zero. Hence,
(\ref{thm:eq-case-degenarate}) holds true.

To prove (\ref{thm:non-case-degenarate}), one uses again
(\ref{eq:expand-any-zero}) taking into account
Corollary~\ref{cor:basis-vector}. Then (\ref{eq:c-k-expand-any-zero})
shows that the only vector polynomial in the zero class is the zero polynomial.
\end{proof}

\begin{remark}
  \label{rem:infinite-int-prob}
  The assertion \ref{thm:case-degenarate}) of
  Theorem~\ref{thm:non-degenerate-case} can be interpreted as
  follows. If the spectral function of $\mathcal{A}$ has a countable
  set of growth points not accumulating anywhere, then the spectrum of
  the operator of multiplication consists only of eigenvalues which,
  due to the fact that $\sigma$ is an $n\times n$ matrix, have
  multiplicity not greater than $n$. Let $\{x_l\}_{l=1}^\infty$ be the
  eigenvalues of the multiplication operator by the independent
  variable in $L_2(\reals,\sigma)$ enumerated taking into account their multiplicity. Hence
the vector polynomials
  $\{\pb{q}_j\}_{j=1}^{j_0}$ are generators of the interpolation
  problem
  \begin{equation}
    \label{eq:inner-int-prob}
    \inner{\pb{r}(x_l)}{\sigma_l\pb{r}(x_l)}_{\complex^n}=0\,,\quad l\in\nats,
  \end{equation}
  where $\sigma_l$ is a matrix of the same form as
  r.\,h.\,s. of (\ref{eq:spec-func-form}) and has the same properties.
  Note that (\ref{eq:inner-int-prob}) yields a linear interpolation
  problem with an infinite set of nodes of interpolation.
\end{remark}

\section{Spectral functions in the self-adjoint case}
\label{sec:Spectral functions self-adjoint}
The operator $A$ is symmetric and, by definition, closed. In this
section, we are interested in the case when $A=A^*$. So let us touch
upon some criteria for self-adjointness of $A$.

Our first criterion is based on the fact that any semi-infinite band
matrix can be considered as a block semi-infinite Jacobi
matrix. Indeed, any semi-infinite band matrix with $2n+1$ diagonals is
equivalent to a semi-infinite Jacobi matrix where each entry is a
$p\times p$ matrix with $p\ge n$. Since the operator $A^*$ is the
operator defined by the matrix $\mathcal{A}$ in the maximal domain
\cite[Sec.\,47]{MR1255973}, the fact that the operator $A$ is
self-adjoint depends exclusively on the asymptotic behavior of the
diagonal elements of its matrix representation $\mathcal{A}$. For any
matrix in $\mathcal{M}(n,\infty)$, consider the semi-infintie
submatrix after the last degeneration, which we called the ``tail of
the matrix'' (see Remark~\ref{rem:tail-matrix}). This ``tail'' can be
seen as a semi-infinite block Jacobi matrix.
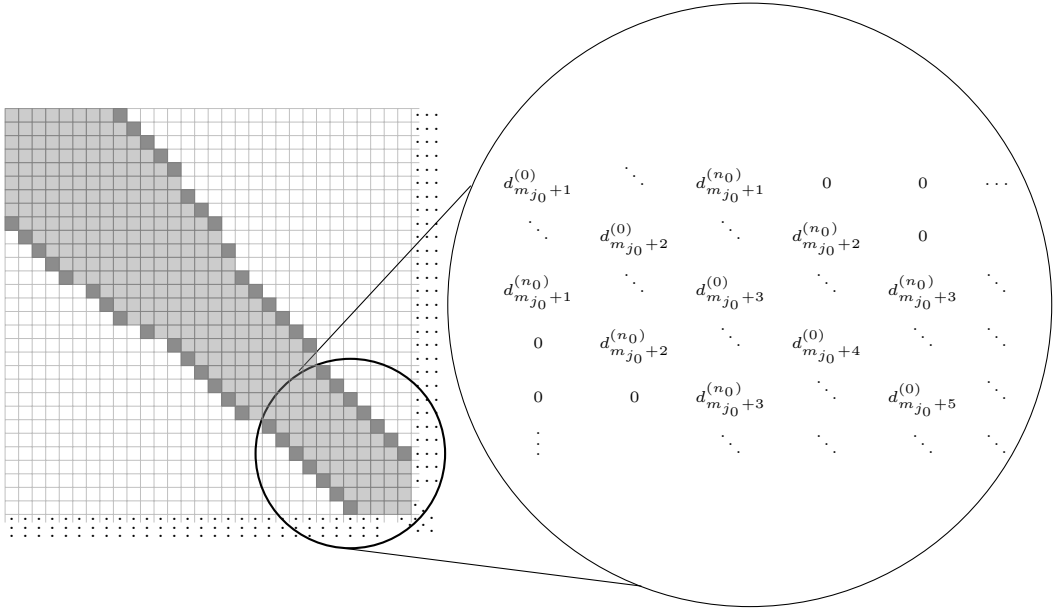
\begin{figure}[h]
\begin{center}
\begin{tikzpicture}[scale=.18,ball/.style = 
{circle, draw, align=center, anchor=north, inner sep=0}]\footnotesize
 \pgfmathsetmacro{\xone}{0}
 \pgfmathsetmacro{\xtwo}{ 30.6}
 \pgfmathsetmacro{\yone}{-0.6}
 \pgfmathsetmacro{\ytwo}{30}
  \draw[step=1cm,gray,opacity=0.5,very thin] (\xone,\yone) grid (\xtwo,\ytwo);
\draw[thick] (25.5,4.5) circle(7);
\draw (21.7,10.6)--(34.4,24.3);
\draw (25,-2.5)--(47,-5.3);

\node[ball,text width=6.9cm,anchor=base] (pncan) at (55,15)
   {$  
\begin{tiny}
  \begin{array}{cccccc}
    d_{m_{j_0}+1}^{(0)}&\ddots&d_{m_{j_0}+1}^{(n_0)}&0&0&\cdots\\
    \ddots&d_{m_{j_0}+2}^{(0)}&\ddots&d_{m_{j_0}+2}^{(n_0)}&0&\\
    d_{m_{j_0}+1}^{(n_0)}&\ddots&d_{m_{j_0}+3}^{(0)}&\ddots&d_{m_{j_0}+3}^{(n_0)}&\ddots\\
    0&d_{m_{j_0}+2}^{(n_0)}&\ddots&d_{m_{j_0}+4}^{(0)}&\ddots&\ddots\\
    0&0&d_{m_{j_0}+3}^{(n_0)}&\ddots&d_{m_{j_0}+5}^{(0)}&\ddots\\
    \vdots&&\ddots&\ddots&\ddots&\ddots
  \end{array}\end{tiny}
$};
 
\draw(30.9,-.3)node[scale=.8]{$\ddots$};
\draw(30.9,.7)node[scale=.8]{$\ddots$};
\draw(29.9,-.3)node[scale=.8]{$\ddots$};
\begin{scope}
\foreach \x in {0,1,2,3,4,5,6,7,8,9,10,11,12,13,14,15,16,
17,18,19,20,21,22,23,24,25,26,27}
\draw(31.3,2.5+\x)node[scale=.8]{$\dots$}
;
\end{scope}

\begin{scope}
\foreach \x in {0,1,2,3,4,5,6,7,8,9,10,11,12,13,14,15,16,17,
18,19,20,21,22,23,24,25,26,27}
\draw(\x+.5,-.5)node[scale=.8]{$\vdots$}
;
\end{scope}
\begin{scope}
\foreach \x in {0,...,4}
{
  \filldraw[thin,gray,opacity=.9] (0+\x, 21-\x)
    rectangle (1+\x,22-\x)
 ;
   \filldraw[thin,gray,opacity=.9] (8+\x, 30-\x)
     rectangle (9+\x,29-\x);}
\end{scope}
\begin{scope}
\foreach \x in {0,...,5}
{
  \filldraw[thin,gray,opacity=.4] (0+\x, 22-\x)
    rectangle (1+\x,23-\x)
 ;
   \filldraw[thin,gray,opacity=.4] (7+\x, 30-\x)
     rectangle (8+\x,29-\x);}
\end{scope}

\begin{scope}
\foreach \x in {6,...,8}
{
  \filldraw[thin,gray,opacity=.9] (0+\x, 22-\x)
    rectangle (1+\x,23-\x)
 ;
   \filldraw[thin,gray,opacity=.9] (7+\x, 30-\x)
     rectangle (8+\x,29-\x);}
\end{scope}
\begin{scope}
\foreach \x in {0,...,9}
{
  \filldraw[thin,gray,opacity=.4] (0+\x, 23-\x)
    rectangle (1+\x,24-\x)
 ;
   \filldraw[thin,gray,opacity=.4] (6+\x, 30-\x)
     rectangle (7+\x,29-\x);}
\end{scope}
\begin{scope}
\foreach \x in {10}
{
  \filldraw[thin,gray,opacity=.9] (0+\x, 23-\x)
    rectangle (1+\x,24-\x)
 ;
   \filldraw[thin,gray,opacity=.9] (6+\x, 30-\x)
     rectangle (7+\x,29-\x);}
\end{scope}
\begin{scope}
\foreach \x in {0,...,11}
{
  \filldraw[thin,gray,opacity=.4] (0+\x, 24-\x)
    rectangle (1+\x,25-\x)
 ;
   \filldraw[thin,gray,opacity=.4] (5+\x, 30-\x)
     rectangle (6+\x,29-\x);}
\end{scope}
\begin{scope}
\foreach \x in {12,...,17}
{
  \filldraw[thin,gray,opacity=.9] (0+\x, 24-\x)
    rectangle (1+\x,25-\x)
 ;
   \filldraw[thin,gray,opacity=.9] (5+\x, 30-\x)
     rectangle (6+\x,29-\x);}
\end{scope}

\begin{scope}
\foreach \x in {0,...,18}
{
  \filldraw[thin,gray,opacity=.4] (0+\x, 25-\x)
    rectangle (1+\x,26-\x)
 ;
   \filldraw[thin,gray,opacity=.4] (4+\x, 30-\x)
     rectangle (5+\x,29-\x);}
\end{scope}

\begin{scope}
\foreach \x in {19,...,25}
{
  \filldraw[thin,gray,opacity=.9] (0+\x, 25-\x)
    rectangle (1+\x,26-\x)
 ;
   \filldraw[thin,gray,opacity=.9] (4+\x, 30-\x)
     rectangle (5+\x,29-\x);}
\end{scope}
\begin{scope}
\foreach \x in {0,...,26}
{
  \filldraw[thin,gray,opacity=.4] (0+\x, 26-\x)
    rectangle (1+\x,27-\x)
 ;
   \filldraw[thin,gray,opacity=.4] (3+\x, 30-\x)
     rectangle (4+\x,29-\x);}
\end{scope}

\begin{scope}
\foreach \x in {}
{
  \filldraw[thin,gray,opacity=.9] (0+\x, 26-\x)
    rectangle (1+\x,27-\x)
 ;
   \filldraw[thin,gray,opacity=.9] (3+\x, 30-\x)
     rectangle (4+\x,29-\x);}
\end{scope}

\begin{scope}
\foreach \x in {0,...,27}
{
  \filldraw[thin,gray,opacity=.4] (0+\x, 27-\x)
    rectangle (1+\x,28-\x)
 ;
   \filldraw[thin,gray,opacity=.4] (2+\x, 30-\x)
     rectangle (3+\x,29-\x);}
\end{scope}

\begin{scope}
\foreach \x in {}
{
  \filldraw[thin,gray,opacity=.9] (0+\x, 27-\x)
    rectangle (1+\x,28-\x)
 ;
   \filldraw[thin,gray,opacity=.9] (2+\x, 30-\x)
     rectangle (3+\x,29-\x);}
\end{scope}
\begin{scope}
\foreach \x in {0,...,28}
{
  \filldraw[thin,gray,opacity=.4] (0+\x, 28-\x)
    rectangle (1+\x,29-\x)
 ;
   \filldraw[thin,gray,opacity=.4] (1+\x, 30-\x)
     rectangle (2+\x,29-\x);}
\end{scope}

\begin{scope}
\foreach \x in {}
{
  \filldraw[thin,gray,opacity=.9] (0+\x, 28-\x)
    rectangle (1+\x,29-\x)
 ;
   \filldraw[thin,gray,opacity=.9] (1+\x, 30-\x)
     rectangle (2+\x,29-\x);}
\end{scope}

\begin{scope}
  \foreach \x in
  {0,...,29}
  {\filldraw[thin,gray,opacity=.25] (0+\x, 29-\x) rectangle
    (1+\x,30-\x) ; \filldraw[thin,gray,opacity=.2] (0+\x, 29-\x)
    rectangle (1+\x,30-\x);}
\end{scope}
\end{tikzpicture}
\end{center}
\end{figure}

Let us denote
\begin{equation*}
\begin{pmatrix}
Q_1&B_1^*&0&0&\cdots\\
B_1&Q_2&B_2^*&0&\\
0&B_2&Q_3&B_3^*&\ddots\\
0&0&B_3&Q_4&\ddots\\
\vdots&&\ddots&\ddots&\ddots
\end{pmatrix}:=
\begin{pmatrix}
d_{m_{j_0}+1}^{(0)}&\ddots&d_{m_{j_0}+1}^{(n_0)}&0&0&\cdots\\
\ddots&d_{m_{j_0}+2}^{(0)}&\ddots&d_{m_{j_0}+2}^{(n_0)}&0&\\
d_{m_{j_0}+1}^{(n_0)}&\ddots&d_{m_{j_0}+3}^{(0)}&\ddots&
d_{m_{j_0}+3}^{(n_0)}&\ddots\\
0&d_{m_{j_0}+2}^{(n_0)}&\ddots&d_{m_{j_0}+4}^{(0)}&\ddots&\ddots\\
0&0&d_{m_{j_0}+3}^{(n_0)}&\ddots&d_{m_{j_0}+5}^{(0)}&\ddots\\
\vdots&&\ddots&\ddots&\ddots&\ddots
\end{pmatrix}\,,
\end{equation*}
where each entry is an $n_0\times n_0$ matrix ($n_0:=n-j_0$). Clearly, the elements
of the block diagonal adjacent to the main diagonal, i.\,e., the
matrices $\{B_k\}_{k\in\nats}$ and $\{B_k^*\}_{k\in\nats}$,
are upper and, respectively, lower triangular matrices such that the
main diagonal entries are positive numbers.

The following proposition is the analogue of the Carleman criterion
\cite[Chap.\,1, Addenda and Problems]{MR0184042}
for block Jacobi matrices.
\begin{proposition}(\cite[Ch.~7,~Thm.~2.9]{MR0222718})
  \label{prop:Carleman-berezanskii}
  If $\sum_{j=0}1/\norm{B_j}$ diverges, then $A$ is self-adjoint.
\end{proposition}
In \cite[Cor.~2.5]{MR1711874}, the following necessary conditions for
self-adjointness are given.
\begin{proposition}
  \label{prop:criteria-mirzoev}
  Suppose that, starting from some $k_0$, all the matrices $Q_{k}$ are
  invertble. If
  \begin{equation*}
    \lim_{k\rightarrow+\infty}\norm{Q_{k}^{-1}}=0\,,
\text{ and }\lim\sup_{k\rightarrow +\infty}
\{\norm{Q_{k}^{-1}B_{k}}+\norm{Q_{k}^{-1}B^*_{k}}\}<1\,,
\end{equation*}
then the operator  $A$ is self-adjoint.
\end{proposition}

Another criterion is given by perturbation theory. Indeed, consider
the operators $D_j$ $(j=0,1,\dots,n)$, whose matrix representation
with respect to $\{\delta_k\}_{k\in\mathbb{N}}$ is a diagonal matrix,
i.\,e., $D_j\delta_k=d^{(j)}_k\delta_k$ for all $k\in\mathbb{N}$,
where $d_k^{(j)}$ is a real number (see
\cite[Sec.~47]{MR1255973}). Note that $\mathcal{D}_j$, given in the
Introduction, is the matrix representation of the operator $D_j$ with
respect to $\{\delta_k\}_{k\in\mathbb{N}}$.  Define the shift operator
$S$ as follows
\begin{equation*}
  S\delta_k=
    \delta_{k+1}\,,\quad \text{for all } k\in\mathbb{N}\,,
\end{equation*}
where by linearity, it is defined on $\Span\{\delta_k\}_{k=1}^\infty$
and then extended to $\cH$ by continuity. Consider the symmetric operator
\begin{equation}
  \label{eq:def_of_operator}
  A':=D_0+\sum_{j=1}^nS^jD_j + \sum_{j=1}^nD_j(S^*)^j\,.
\end{equation}
Now, if the operator $\sum_{j=1}^nS^jD_j+\sum_{j=1}^nD_j(S^*)^j$ is
$D_0$-bounded with $D_0$-bound smaller than $1$ (see
\cite[Sec.\,5.1]{MR566954}), one can resort to the Rellich-Kato
theorem \cite[Thm.\,4.3]{MR0407617} to show that $A'$ is
self-adjoint. When this happens, it can be shown that $A=A'$.

Let us assume from this point to the end of this section that the
operator $A$ is self-adjoint. Our approach to constructing the
spectral functions of $A$ is based on techniques of perturbation
theory related with the strong resolvent convergence (see
\cite[Sec.\,9.3]{MR566954}).

We begin by recalling the following definition.
\begin{definition}
\label{def:core}
A subset $D$ of the domain of a closeable operator $B$ is called
a core of $B$ when $\overline{B\upharpoonright_{D}}=B$.
\end{definition}
Also, we recur to the following known results
(cf. \cite[Chap.\,8, Cor.\,1.6 and Thm.\,1.15]{MR0407617}):
\begin{proposition}\cite[Thm.\,9.16]{MR566954}.
\label{prop:weidmann}
Let $\{B_N\}_{N\in\mathbb{N}}$ and $B$ be self-adjoint operators on
$\cH$. If there is a core $D$ of $B$ such that for every $\varphi\in
D$ there is an $N_0\in\mathbb{N}$ which satisfies $\varphi\in\dom(B_N)$
for $N\ge N_0$ and $B_N\varphi\rightarrow B\varphi$, then the sequence
$\{(B_N-zI)^{-1}\}_{N\in\mathbb{N}}$ converges strongly to
$(B-zI)^{-1}$ (denoted
$(B_N-zI)^{-1}\xrightarrow[N\rightarrow\infty]{s}(B-zI)^{-1}$) for all
$z\in\complex\setminus\reals$.
\end{proposition}

\begin{proposition}\cite[Thm.\,9.19]{MR566954}.
  \label{prop:convergence-spectral-function}
  Let $\{B_N\}_{N\in\mathbb{N}}$ and $B$ be self-adjoint operators on
  $\cH$, such that the sequence $\{(B_N-iI)^{-1}\}_{N\in\mathbb{N}}$
  converges strongly to $(B-iI)^{-1}$. Then
  \begin{equation*}
    \begin{aligned}
      E_{B_N}(t)\xrightarrow[N\rightarrow\infty]{s} E_{B}(t)\\
      E_{B_N}(t+0)\xrightarrow[N\rightarrow\infty]{s} E_{B}(t)
    \end{aligned}\,,
    \quad \text{for all}\ \ t\in\reals\ \ \text{such that}\ \ E_B(t)=E_B(t+0)\,.
  \end{equation*}
Here, $E_{B_N}(t)$ and $E_{B}(t)$ are the spectral
resolutions of the identity of $B_N$ and $B$, respectively.
\end{proposition}

Recall the finite dimensional operator $\widetilde{A}_N$ studied
in Section~\ref{sec:submatrices} and define
\begin{equation*}
  A_N:=\widetilde{A}_N\oplus\mathbb{O}\,,
\end{equation*}
where $\mathbb{O}$ is the zero-operator in the infinite dimensional
space $\cH\ominus\cH_N$. For any $N>n$, the operator $A_N$ is
self-adjoint, so we take advantage of the spectral theorem. Let us
introduce the following notation for the matrix valued spectral
functions
\begin{align}
\label{eq:sigma-N-sa}
\sigma_N(t)&:=\{\inner{\delta_i}{E_{A_N}(t)\delta_j}\}_{i,j=1}^\infty\quad\text{ for any } N>n\\
\label{eq:sigma-sa}
\sigma(t)&:=\{\inner{\delta_i}{E_A(t)\delta_j}\}_{i,j=1}^\infty\,.
\end{align}

 \begin{lemma}
   \label{lem:sepctral-function-convergence}
   The matrix valued functions $\sigma_N(t)$ given in (\ref{eq:sigma-N-sa})
   converge to the matrix valued function $\sigma(t)$, defined by
   (\ref{eq:sigma-sa}), at all points of continuity of $\sigma(t)$, i.\,e.,
   \begin{equation}
\label{eq:converges-sigmas}
     \sigma_N(t)\xrightarrow[N\rightarrow\infty]{}\sigma(t)\,,\quad t
     \text{ being a point of continuity of } \sigma(t)\,.
   \end{equation}
 \end{lemma}
 \begin{proof}
   Let $l_{\rm{fin}}(\nats)$ be the linear space of sequences with a
   finite number of nonzero elements. This space is a core of the
   operator $A$. Given an element
   $\varphi=\sum_{k=1}^{s}\varphi_k\delta_k\in l_{\text{fin}}(\nats)$, one
   verifies that, for all $N\geq N_0=s+n$,
   $A_N\varphi=A\varphi$. Therefore, the conditions of
   Proposition~\ref{prop:weidmann} are satisfied. So, by
   Proposition~\ref{prop:convergence-spectral-function}, one obtains
   the result.
 \end{proof}
 \begin{corollary}
   \label{cor:moment_measure_without_t}
   For any $k\in\nats\cup\{0\}$, the integral
   \begin{equation*}
     \int_\reals t^kd\sigma
   \end{equation*}
   converges. Moreover, $\int_\reals d\sigma$ is the identity matrix.
 \end{corollary}
 \begin{proof}
   The first part of the assertion is a consequence of
   Lemmas~\ref{lem:subsequence-converges-integral} and
   \ref{lem:sepctral-function-convergence}. The second part follows
   from the fact that $\sigma$ is the spectral function of the
   self-adjoint operator $A$.
 \end{proof}
On the basis of the previous result, let us denote
\begin{equation*}
  S_k:=\int_\reals t^kd\sigma
\end{equation*}
for any $k\in\nats\cup\{0\}$.
\begin{definition}
  \label{def:T-spectral-measures}
  Given the spectral function $\sigma$ of the self-adjoint operator
  $A$, denote
  \begin{equation*}
    \sigma_{\mathscr{T}}:=\mathscr{T}\sigma\mathscr{T}^*\,,
  \end{equation*}
  where $\mathscr{T}$ is a matrix satisfying \ref{item:prop-T}) and
  \ref{item:prop-T-2}) (see \eqref{eq:initial-condition}).
\end{definition}
Using (\ref{eq:finite-measure-finite-first-moment}), one obtains
\begin{equation}
\label{eq:convergence-measures-sa}
 \sigma_N^\mathscr{T}(t)\xrightarrow[N\rightarrow\infty]{}
\sigma_\mathscr{T}(t)\,,\quad \text{for }t 
\text{ being a point of continuity of } \sigma(t)\,,
\end{equation}
where $\sigma^{\mathscr{T}}_N$ is the function given in (\ref{eq:spec-func-form}).
It also holds that
\begin{equation}
\label{eq:infinite-moment-for-t}
  \mathscr{T}S_k\mathscr{T}^*=\int_{\reals}t^k\,d\sigma_{\mathscr{T}}=:
\widetilde{S}_k(\mathscr{T}).
\end{equation}

\begin{lemma}
  \label{lem:moment-problem}
 For any matrix $\mathscr{T}$ satisfying \ref{item:prop-T}) and
  \ref{item:prop-T-2}), the
 function $\sigma_{\mathscr{T}}$, given in
 Definition~\ref{def:T-spectral-measures}, is in
 $\mathfrak{M}_d(n,\infty)$ (see Definition~\ref{def:class-sigma-infinite}).
\end{lemma}
\begin{proof}
  It follows from \cite[Sec.\,2]{MR1882637} that the sequence
  $\{S_k\}_{k=0}^\infty$ defines a determinate moment problem. In view
  of (\ref{eq:infinite-moment-for-t}), the sequence
  $\{\widetilde{S}_k(\mathscr{T})\}_{k=0}^\infty$ also has only one solution for
  any $\mathscr{T}$.
\end{proof}


\section{Reconstruction of the matrix}
\label{sec:Reconstruction}
In this section, the starting point is a matrix valued function
$\widetilde{\sigma}\in\mathfrak{M}(n,\infty)$
(see Definition~\ref{def:class-sigma-infinite}) and we construct a
matrix $\mathcal{A}$
in the class $\mathcal{M}(n,\infty)$
from this function. Furthermore, we verify that, for some matrix
$\mathscr{T}$
which gives the initial conditions (see \eqref{eq:initial-condition}),
$\widetilde{\sigma}$
is the spectral function of the reconstructed matrix $\mathcal{A}$.
Hence, any matrix in $\mathcal{M}(n,\infty)$
can be reconstructed from its function in $\mathfrak{M}(n,\infty)$.

Consider the Hilbert space $L_2(\reals,\widetilde{\sigma})$ with
$\widetilde{\sigma}\in\mathfrak{M}(n,\infty)$. From what has been said, either
there are polynomials of zero norm in this space or there are not.
Let us apply the Gram-Schmidt procedure of orthonormalization to the
sequence of vector polynomials given by
\eqref{eq:canonical-vector-polynomial}. If there exist nozero polynomials
whose norm is zero, then the Gram-Schmidt algorithm yields
vector polynomials of zero norm. Indeed, let $\pb{r}\not\equiv 0$ be a vector
polynomial of zero norm of minimal height $h_1$ (that is, any nonzero
polynomial of zero norm has height no less than $h_1$), and let
$\{\tb{p}_k\}_{k=1}^{h_1}$ be the orthonormalized vector polynomials
obtained by the first $h_1$ iterations of the Gram-Schmidt
procedure. Hence, if one defines
\begin{equation*}
  \pb{s}=\pb{e}_{h_1+1}-\sum_{i=1}^{h_1}\inner{\tb{p}_i}{\pb{e}_{h_1+1}}\tb{p}_i\,,
\end{equation*}
then, in view of the fact that $h(\tb{p}_k)=k-1$ for $k=1,\dots, h_1$,
and taking into account \cite[Thm.\,2.1]{2014arXiv1401.5384K},
one has
\begin{equation*}
\pb{e}_{h_1+1}=a\pb{r}+\sum_{i=1}^{h_1}a_i\tb{p}_i
\end{equation*}
which in turn leads to
\begin{equation}
\label{eq:gram-schmidt-vector}
\pb{s}=a\pb{r}+\sum_{k=1}^{h_1}\widetilde{a}_k\tb{p}_k\,.
\end{equation}
This implies that $\norm{\pb{s}}_{L_2(\reals,\widetilde{\sigma})}=0$
since $\inner{\pb{s}}{\pb{r}}_{L_2(\reals,\widetilde{\sigma})}=0$ and
$\pb{s}\perp\tb{p}_k$ for $k=1,\dots,h_1$ by construction. Thus, the
Gram-Schmidt procedure yields vector polynomials of zero norm.

Having found a vector polynomials of zero norm,
one continues with the procedure taking the next vector of the
sequence \eqref{eq:canonical-vector-polynomial}.  Observe that if the
Gram-Schmidt technique has produced a vector polynomial of zero norm
$\pb{q}$ of height $h$,
then for any integer number $l$, the vector polynomial $\pb{t}$ that
is obtained at the $h+1+nl$-th iteration of the Gram-Schmidt
process, viz.,
\begin{equation*}
  \pb{t}=\pb{e}_{h+1+nl}-\sum_{h(\tb{p}_i)<h+nl}
\inner{\tb{p}_i}{\pb{e}_{h+1}}\tb{p}_i\,,
\end{equation*}
satisfies that $\norm{\pb{t}}_{L_2(\reals,\widetilde{\sigma})}=0$ (for all
$l\in\nats$). Indeed, due to \cite[Thm.\,2.1]{2014arXiv1401.5384K},
one has
\begin{equation*}
\pb{e}_{h+1+nl}=z^l\pb{q}+\sum\limits_{h(\tb{p}_i)<h+nl} 
c_i\tb{p}_i+\sum\limits_{h(\pb{r})<h+nl} \pb{r}\,,
\end{equation*}
where each $\pb{r}$
is a vector polynomial of zero norm obtained from the Gram-Schmidt procedure.
\begin{remark}
\label{rem:infinite-support}
Since $\widetilde{\sigma}$ has an infinite number of growth points,
the HIlbert space $L_2(\reals,\widetilde{\sigma})$ is infinite
dimensional \cite[Sec.\,72]{MR1255973}. Thus,
the Gram-Schmidt procedure renders an infinite sequence of orthonormal
vectors.
\end{remark}


The following flow chart shows that the Gram-Schmidt procedure applied
to the sequence \eqref{eq:canonical-vector-polynomial} gives not only
the orthonormalized sequence, but also a sequence of null vector
polynomials such that at any step of the algorithm these two sequences
together are a basis of the space of vector polynomials (see
\cite[Thm.\,2.1]{2014arXiv1401.5384K} and compare with
\eqref{eq:basis-p-q}).

\tikzstyle{decision} = [diamond, draw, text width=7.6em, text badly
centered, node distance=4cm, inner sep=0pt, aspect=2, rounded corners]
\tikzstyle{decision1} = [diamond, draw, text width=4.5em, text badly
centered, node distance=4cm, inner sep=0pt, aspect=2, rounded corners]
\tikzstyle{block} = [rectangle, draw, text width=5em, text centered,
rounded corners, minimum height=3em] \tikzstyle{block1} = [rectangle,
draw, text width=8em, text centered, rounded corners, minimum
height=3em] \tikzstyle{line} = [draw, -latex'] \tikzstyle{cloud} =
[draw, ellipse, node distance=3cm, minimum height=2em, text width=3em,
text centered]
\begin{small}
  \begin{center}
    \begin{tikzpicture}[node distance = 3cm, auto]
      \node [block] (init) {$\widehat{\pb{p}}_k=\pb{e}_k$}; \node
      [cloud, left of=init] (initial-condition) {$s:=0$, $j:=1$,
        $k:=1$}; \node [block1, below of=init, node distance=2cm]
      (normalized)
      {$\tb{p}_{k-(j+s-1)}=\frac{\widehat{\pb{p}}_k}{\norm{\widehat{\pb{p}}_k}}$};
      \node [block, below of=normalized, node distance=2cm]
      (k-counter) {$k=k+1$}; \node [block, right of=k-counter, node
      distance=4.5cm] (s-counter) {$s=s+1$}; \node [block, left
      of=k-counter] (j-counter) {$j=j+1$}; \node [decision1, below
      of=k-counter, node distance=2.6cm] (j-1) {$j=1$?}; \node
      [decision, right of=j-1, node distance=4.5cm] (k-mod-n) {
        $k=h(\tb{q}_{i})+nl$?, $l\!\in\!\nats,\,i\!=\!1,\dots,j$};
      \node [block, left of=j-1, node distance=3cm] (q-poly)
      {$\tb{q}_j:=\widehat{\pb{p}}_k$}; \node [block1, below of=j-1,
      node distance=2.1cm, text width=10.5em] (G-S)
      {$\widehat{\pb{p}}_k=\pb{e}_k-\sum\limits_{i<k}
        \inner{\tb{p}_i}{\pb{e}_{k}}\tb{p}_i$}; \node [decision, below
      of=G-S,, node distance=2.2cm] (zero-norm) {
        $\norm{\widehat{\pb{p}}_k}=0$?};
      \path [line] (initial-condition) -- (init); \path [line] (init)
      -- (normalized); \path [line] (normalized) -- (k-counter); \path
      [line] (k-counter) -- (j-1); \path [line] (j-1) -- node [near
      start] {yes} (G-S); \path [line] (G-S) -- (zero-norm); \path
      [line] (zero-norm.east) -- node [near start] {no} ++(5.4cm,0) |-
      (normalized); \path [line] (zero-norm) -| node [near start]
      {yes} (q-poly); \path [line] (q-poly) -- (j-counter); \path
      [line] (s-counter)--(k-counter); \path [line]
      (j-counter)--(k-counter); \path [line] (j-1)-- node [near start]
      {no} (k-mod-n); \path [line] (k-mod-n)-- node [near start] {yes}
      (s-counter); \path [line] (k-mod-n)|- node [near start] {no}
      (G-S);
    \end{tikzpicture}
  \end{center}
\end{small}

Clearly, since the support of the measure is infinite according to
Definition~\ref{def:class-sigma-infinite}, one cannot obtain more than $n-1$
null vectors from the Gram-Schmidt procedure applied to the sequence
of vector polynomials given by
\eqref{eq:canonical-vector-polynomial}. Indeed, if one finds the
$n$-th vector polynomial $\tb{q}_n$, by repeating the argument
described above and taking into account
\begin{equation*}
  \{h(\tb{q}_1),\dots,h(\tb{q}_n)\}=\integers/n\integers\,,
\end{equation*}
one obtains that all the vectors provided by this procedure have zero
norm beginning from some vector. This would correspond to an infinite loop
in the left side of the flow chart and to a measure with finite support
since  $L_2(\reals,\widetilde{\sigma})$ would be finite
dimensional.
\begin{lemma}
  \label{lem:basis-throught-gram-schmidt}
  Any vector polynomial $\boldsymbol{r}(z)$ is a finite linear
  combination of
\begin{equation}
\label{eq:basis-gram-schimdt}
  \{\tb{p}_k(z): k\in\nats\} \cup \{z^l\tb{q}_j(z): 
l\in\nats\cup\{0\},\, j\in\{1,\dots,j_0\}\}\,.  
\end{equation}
\end{lemma}
\begin{proof}
  Note that the vector polynomials defined in
  (\ref{eq:canonical-vector-polynomial}) satisfy that
  $h(\pb{e}_i)=i-1$. 
 Due to the fact that
  \begin{equation}
    \label{eq:GramScmidt-procedure}
    h\left(\pb{e}_k-\!\!\!\!\sum_{h(\tb{p}_i)<k-1}\!\!
\inner{\tb{p}_i}{\pb{e}_k}\tb{p}_i\right)
    =h(\pb{e}_k)\,,
  \end{equation}
  one concludes that the heights of the set
  $\{\tb{p}_k(z)\}_{k=1}^N\cup\{z^l\tb{q}_i(z)\}_{i=1}^n$
  ($l\in\nats\cup\{0\}$) are in one-to-one correspondence with the set
  $\nats\cup\{0\}$.
  To complete the proof, it only remains to use
  \cite[Thm.\,2.1]{2014arXiv1401.5384K}.
\end{proof}

By the argumentation given above and the same resoning used in the
proof of
Theorem~\ref{thm:non-degenerate-case}~\ref{thm:case-degenarate}), one
arrives at the following assertion.
\begin{proposition}
  \label{prop:interpolation-infinite}
  Let $\widetilde{\sigma}$ be in $\mathfrak{M}(n,\infty)$. There exist
  at most $n-1$ vector polinomials $\{\tb{q}_i\}_{i=1}^{j_0}$
  ($j_0\le n-1$) such that any vector polynomial $\pb{r}$ of zero norm
  can be written as
  \begin{equation*}
    \pb{r}=\sum_{i=1}^{j_0}R_i\tb{q}_i\,,
  \end{equation*}
where $R_i$, for any $i\in\{1,\dots,j_0\}$, is a scalar polynomial.
\end{proposition}

Let $\widetilde{\sigma}(t)$ be a matrix valued function in
$\mathfrak{M}(n,\infty)$
and consider the sequences $\{\tb{p}_k\}_{k\in\nats}$ and
$\{\tb{q_i}\}_{i=1}^{j_0}$ obtained by applying the Gram-Schimdt
process to the sequence \eqref{eq:canonical-vector-polynomial}.
Since for any $k\in\nats$ there exists $l\in\nats$ such that
$h(z\tb{p}_k)\le h(\tb{p}_l)$, one has by
Lemma~\ref{lem:basis-throught-gram-schmidt} that
  \begin{equation}
    \label{eq:representation-z-q}
    z\widetilde{\boldsymbol{p}}_k(z)=
\sum_{i=1}^{l}c_{ik}\widetilde{\boldsymbol{p}}_i(z)+
\sum_{j=1}^{j_0}R_{kj}(z)\tb{q}_j(z)\,,
  \end{equation}
where $c_{ik}\in\mathbb{C}$ and $R_{kj}(z)$ is a scalar
polynomial.

\begin{remark}
\label{rem:items-for-coefficients}
By comparing the heights of the left and right hand sides of
(\ref{eq:representation-z-q}), one obtains the following relations
given in items
~\ref{item:heigt-equal-zp-p1}) and \ref{item:heigt-equal-zp-p2}) below. To
verify item~\ref{item:heigt-equal-zp-p}), one has to take into account
that the leading coefficient of $\pb{e}_k$ is positive for $k\in\nats$
and therefore the Gram-Schmidt procedure yields the sequence
$\{\tb{p}_k\}_{k=1}^{\infty}$ with its elements having positive
leading coefficients (cf. \cite[Rem.\,4]{2014arXiv1409.3868K}).
\begin{enumerate}[$i$)]
\item $c_{lk}=0$ if $h(z\tb{p}_k)<h(\tb{p}_l)$\label{item:heigt-equal-zp-p1},
\item $R_{kj}(z)=0$ if $h(z\tb{p}_k)<h(R_{kj}(z)\tb{q}_j)$\label{item:heigt-equal-zp-p2},
\item $c_{lk}>0$ if there is $l\in\nats$ such that $h(z\tb{p}_k)=h(\tb{p}_l)$\label{item:heigt-equal-zp-p}.
\end{enumerate}
\end{remark}
Clearly (recall that our inner product is antilinear in its first argument),
\begin{equation}
\label{eq:def-c-ki-inner}
  c_{lk}=
\inner{\widetilde{\boldsymbol{p}}_l}{z\widetilde{\boldsymbol{p}}_k}_{L_2(\mathbb{R},\widetilde{\sigma})}
=\inner{z\widetilde{\boldsymbol{p}_l}}{\widetilde{\boldsymbol{p}}_k}_{L_2(\mathbb{R},\widetilde{\sigma})}
=c_{kl}\,.
\end{equation}
In \cite[Sec.\,3]{2014arXiv1409.3868K}, a
reconstruction algorithm is provided for recovering the finite band
matrix associated to the operator $A_N$ from its spectral
function. The proof of \cite[Lem.\,4.1]{2014arXiv1409.3868K} proves the
following assertion
\begin{proposition}
  \label{prop:band-matrix}
If $\abs{l-k}>n$. Then, the complex numbers $c_{ki}$ in
(\ref{eq:representation-z-q}) obey
\begin{equation*}
  c_{kl}=c_{lk}=0\,.
\end{equation*}
\end{proposition}

Proposition~\ref{prop:band-matrix} shows that
$\{c_{lk}\}_{l,k=1}^{\infty}$ is a band matrix. Let us turn to the
question of characterizing the diagonals of
$\{c_{lk}\}_{l,k=1}^{\infty}$. It will be shown that they undergo the
kind of degeneration given in the Introduction.

For a fixed number $i\in\{0,\dots,n\}$, we define the numbers
\begin{equation}
  \label{eq:entries-matrix-inverse-prob}
  d^{(i)}_k:=c_{k+i, k}=c_{k, k+i}
\end{equation}
for $k\in\nats$. The proof of the following assertion repeats the one
of \cite[Lem.\,4.2]{2014arXiv1409.3868K}.
\begin{proposition}
  \label{prop:reconstruction-entries}
  Fix $j\in\{0,\dots,j_0-1\}$.
  \begin{enumerate}[$i$)]
  \item If $k$ is such that $ h(\tb{q}_{j})<
    h(z\tb{p}_k)<h(\tb{q}_{j+1})$, then $d_k^{(n-j)}>0$. Here one
    assumes that $h(\boldsymbol{q}_0):=n-1$.
\item If $k$ is such that $h(z\tb{p}_k)\geq h(\tb{q}_{j+1})$, then
  $d_k^{(n-j)}=0$.
  \end{enumerate}
\end{proposition}
\begin{corollary}
  \label{cor:reconstructed-in-class}
  If $c_{ik}$ are the coefficients given in
  (\ref{eq:representation-z-q}), then the matrix
  $\{c_{kl}\}_{k,l=1}^{\infty}$ is in $\mathcal{M}(n,\infty)$ and it
  is the matrix representation of a symmetric restriction of the
  operator of multiplication by the independent variable in
  $L_2(\reals,\widetilde{\sigma})$. (The restriction could be
  improper, i.\,e., the case when the restriction coincides with the
  multiplication operator is not excluded).
\end{corollary}
\begin{proof}
  Taking into account (\ref{eq:entries-matrix-inverse-prob}), it
  follows from Propositions~\ref{prop:band-matrix} and
  \ref{prop:reconstruction-entries} that the matrix
  $\{c_{kl}\}_{k,l=1}^{\infty}$ is in the class
  $\mathcal{M}(n,\infty)$. Now, in view of (\ref{eq:def-c-ki-inner}),
  the operator of multiplication by the independent variable is an
  extension of the minimal closed symmetric operator $B$ in
  $L_2(\reals,\widetilde{\sigma})$ satisfying
  \begin{equation*}
    c_{kl}=\inner{\widetilde{\pb{p}}_k}{B\widetilde{\pb{p}}_l}\,.
  \end{equation*}
\end{proof}


 \begin{theorem}
  \label{thm:sigma-unique}
  Let $\widetilde{\sigma}$ be an element of
  $\mathfrak{M}(n,\infty)$ and $c_{ik}$ be the coefficients given in
  (\ref{eq:representation-z-q}). Then $\widetilde{\sigma}$ is a
  spectral function of the matrix $\{c_{kl}\}_{k,l=1}^{\infty}$
  according to Definition~\ref{def:spectral-measure-gen-case}.
\end{theorem}
\begin{proof}
  Since the recurrence equation for the orthonormal sequence
  $\{\tb{p}_k\}_{k=1}^{\infty}$ and the sequence of polynomials
  $\{\pb{p}_k\}_{k=1}^{\infty}$ are related in the same way as in the
  finite dimensional case (see \cite[Eqs. 2.17 and
  4.15]{2014arXiv1409.3868K}), one can use the argumentation of the
  proofs of \cite[Lem.\,4.3]{2014arXiv1409.3868K}) to obtain
  that the vector polynomials $\{\boldsymbol{p}_k(z)\}_{k=1}^{\infty}$
  and $\{\tb{p}_k(z)\}_{k=1}^{\infty}$ satisfy
\begin{equation}
\label{eq:polynomial_p_inverse_problem}
 \boldsymbol{p}_k(z)=\tb{p}_k(z)+\pb{r}_k(z)\,,
\end{equation}
where $\norm{\pb{r}_k}_{L_2(\reals,\widetilde{\sigma})}=0$.
Analogously, when $j_0\neq 0$ it can also be proven that the vector
polynomials $\{\tb{q}_j(z)\}_{j=1}^{j_0}$ and
$\{\pb{q}_j(z)\}_{j=1}^{j_0}$ satisfy
  \begin{equation}
    \label{eq:dif-q-tilde-q}
    \boldsymbol{q}_j(z)=\sum_{i\leq j}R_i(z)\tb{q}_i(z)\,,\quad R_j\neq0\,,
  \end{equation}
  where $R_i(z)$ are scalar polynomials (see
  \cite[Lem.\,4.4]{2014arXiv1409.3868K}).  Due to
  (\ref{eq:polynomial_p_inverse_problem}) and (\ref{eq:dif-q-tilde-q})
  $\{\pb{p}_k\}_{k=1}^{\infty}$ is an orthonormal sequence in
  $L_2(\reals, \widetilde{\sigma})$ and $\pb{q}_j$ is the equivalence
  class of zero in this space for any $j\in\{1,\dots,j_0\}$.
\end{proof}
\begin{theorem}
  \label{thm:last-one}
  Let $\widetilde{\sigma}$ be in $\mathfrak{M}_d(n,\infty)$ and
    $c_{ik}$ be the coefficients given in
    (\ref{eq:representation-z-q}). Then there exists $\mathscr{T}$
    such that $\sigma_{\mathscr{T}}$, given in
    Definition~\ref{def:T-spectral-measures}, coincides with $\widetilde{\sigma}$.
\end{theorem}
\begin{proof}
  According to Theorem~\ref{thm:sigma-unique}, there is $\mathscr{T}$
  such that the vector polynomials $\{\pb{p}_k\}_{k=1}^\infty$,
  generated by $\{c_{kl}\}_{k,l=1}^\infty$ and $\mathscr{T}$, are
  orthonormal in $L_2(\reals,\widetilde{\sigma})$. Since
  $\widetilde{\sigma}$ is the unique solution of the moment problem
  \begin{equation*}
    \left\{\int_\reals t^kd\widetilde{\sigma}\right\}_{k=0}^\infty\,,
  \end{equation*}
 the orthonormal system $\{\pb{p}_k\}_{k=1}^\infty$ is a basis and
 $\{c_{kl}\}_{k,l=1}^\infty$ is the corresponding matrix representation of the
 operator of multiplication by the independent variable
 \cite[Sec.\,2]{MR1882637}. Let $\sigma$ be the spectral function
 given by (\ref{eq:sigma-sa}) (with $A$ being the operator of
 multiplication by the independent variable and substituting
 $\delta_k$ by $\pb{p}_k$). Also, let $\sigma_{\mathscr{T}}$ be the function defined in
 Definition~\ref{def:T-spectral-measures}. Since the elements of the sequence
 $\{\pb{p}_k\}_{k=1}^\infty$ satisfy the recurrence equations given by
 the matrix $\{c_{kl}\}_{k,l=1}^\infty$ with initial conditions
 $\mathscr{T}$, for any $k,l\in\nats$, there is $N$ sufficiently large
 such that
 \begin{equation*}
   \inner{\pb{p}_k}{\pb{p}_l}_{L_2(\reals,\sigma_N^\mathscr{T})}=\delta_{kl}\,.
 \end{equation*}
Now, from (\ref{eq:convergence-measures-sa}) and
Lemma~\ref{lem:subsequence-converges-integral} applied to the sequence
$\{\sigma_N^{\mathscr{T}}\}_{N>n}$ and the function $\sigma_{\mathscr{T}}$, it follows that
$\widetilde{\sigma}$ and $\sigma_{\mathscr{T}}$ have the same moments.
\end{proof}

\def\cprime{$'$} \def\lfhook#1{\setbox0=\hbox{#1}{\ooalign{\hidewidth
  \lower1.5ex\hbox{'}\hidewidth\crcr\unhbox0}}} \def\cprime{$'$}
  \def\cprime{$'$} \def\cprime{$'$} \def\cprime{$'$} \def\cprime{$'$}
  \def\cprime{$'$} \def\cprime{$'$}

\end{document}